%% file: 2hamIntrinsic.tex
\documentclass[10pt]{article}
\usepackage{amsmath}
\usepackage{amsfonts}
\usepackage{amssymb}
\usepackage{amsthm}
\usepackage{ifpdf}
\usepackage{subfig}
\usepackage{wrapfig}
\usepackage{cite}
\usepackage{clrscode}
\usepackage{comment}
\usepackage{appendix}

\ifpdf

  \usepackage[pdftex]{epsfig}
  \usepackage[pdftex]{hyperref}

\else

    \usepackage[dvips]{epsfig}
    \newcommand{\href}[2]{#2}

\fi

\newif\ifabstract
\newif\iffull

\ifabstract
	\fullfalse
\else
	\fulltrue
\fi

\newtoks\magicAppendix
\magicAppendix={}
\newtoks\magictoks
\newif\iflater
\laterfalse
\ifabstract
\long\def\later#1{\magictoks={#1}%
  \edef\magictodo{\noexpand\magicAppendix={\the\magicAppendix \par
    \the\magictoks%
  }}
  \magictodo}
\long\def\both#1{\magictoks={#1}%
  \edef\magictodo{\noexpand\magicAppendix={\the\magicAppendix \par
    \noexpand\setcounter{theorem-preserve}{\noexpand\arabic{theorem}}%
    \noexpand\setcounter{theorem}{\arabic{theorem}}%
    \noexpand\setcounter{section-preserve}{\noexpand\arabic{section}}%
    \noexpand\setcounter{section}{\arabic{section}}%
	\noexpand\let\noexpand\oldsection=\noexpand\thesection
	\noexpand\def\noexpand\thesection{\thesection}
	\noexpand\let\noexpand\oldlabel=\noexpand\label
	\noexpand\let\noexpand\label=\noexpand\blank
    \the\magictoks%
    \noexpand\setcounter{theorem}{\noexpand\arabic{theorem-preserve}}%
    \noexpand\setcounter{section}{\noexpand\arabic{section-preserve}}%
	\noexpand\let\noexpand\thesection=\noexpand\oldsection
	\noexpand\let\noexpand\label=\noexpand\oldlabel
  }}
  \magictodo
  \the\magictoks}
\else
\long\def\later#1{#1}
\long\def\both#1{#1}
\fi
\long\def\magicappendix{
	\latertrue%
	\the\magicAppendix%
}

\newtheorem{conjecture}{Conjecture}

\newtheorem{theorem}{Theorem}[section]

\newtheorem{corollary}[theorem]{Corollary}
\newtheorem{definition}[theorem]{Definition}
\newcommand{\ta}{\tilde{\alpha}}
\newcommand{\tb}{\tilde{\beta}}
\newcommand{\tg}{\tilde{\gamma}}
\newcommand{\setr}[2]{\left\{\ #1 \ \left|\ #2 \right. \ \right\}}

\newcommand{\pred}{\mathrm{pred}}

\newcommand{\Z}{\mathbb{Z}}
\newcommand{\N}{\mathbb{N}}

\newcommand{\strength}{{\rm str}}

\newcommand{\dom}{{\rm dom} \;}

\newcommand{\res}[1]{\textrm{res}(#1)}
\newcommand{\termasm}[1]{\mathcal{A}_{\Box}[\mathcal{#1}]}
\newcommand{\prodasm}[1]{\mathcal{A}[\mathcal{#1}]}

\newcommand{\calT}{\mathcal{T}}
\newcommand{\lab}{{\rm label}}

\begin{document}
\ifabstract
 \addtolength{\belowcaptionskip}{-6pt}
 \addtolength{\abovecaptionskip}{-8pt}
\fi

\title{The two-handed tile assembly model is not intrinsically universal
}%

\author{
Erik D. Demaine\thanks{Computer Science and Artificial Intelligence Laboratory,
      Massachusetts Institute of Technology,
      32 Vassar St., Cambridge, MA 02139, USA,
      \protect\url{edemaine@mit.edu}}
\and
Matthew J. Patitz\thanks{Department of Computer Science and Computer Engineering, University of Arkansas,
\protect\url{patitz@uark.edu} This author's research was supported in part by National Science Foundation Grants CCF-1117672 and CCF-1422152.}
\and
Trent A. Rogers\thanks{Department of Computer Science and Computer Engineering, University of Arkansas,
\protect\url{tar003@uark.edu}.  This author's research was supported by the National Science Foundation Graduate Research Fellowship Program under Grant No. DGE-1450079, and National Science Foundation grants CCF-1117672 and CCF-1422152.}
\and  Robert T. Schweller\thanks{ Department of Computer Science, University of Texas--Pan American, Edinburg, TX, 78539, USA. \protect\url{rtschweller@utpa.edu}. This author's research was supported in part by National Science Foundation Grant CCF-1117672. }
\and Scott M. Summers\thanks{Department
of Computer Science and Software Engineering, University of Wisconsin--Platteville, Platteville, WI 53818, USA.
\protect\url{summerss@uwplatt.edu}.}
\and Damien Woods\thanks{Computer Science, California Institute of Technology, \protect\url{woods@caltech.edu}. This author's research was supported by National Science Foundation grants 0832824 (The Molecular Programming Project), CCF-1219274, and CCF-1162589.}}


\date{}

\maketitle

\vspace{-4ex}
\begin{abstract}
The well-studied Two-Handed Tile Assembly Model (2HAM) is a model of tile assembly in which pairs of large assemblies can bind, or self-assemble, together. In order to bind, two assemblies must have matching glues that can simultaneously touch each other, and stick together with strength that is at least the temperature $\tau$, where $\tau$ is some fixed positive integer. We ask whether the 2HAM is intrinsically universal, in other words we ask: is there a single universal 2HAM tile set $U$ which can be used to simulate any instance of the model? Our main result is a negative answer to this question. We show that for all $\tau' < \tau$, each temperature-$\tau'$ 2HAM tile system does not simulate at least one temperature-$\tau$ 2HAM tile system. This impossibility result proves that the 2HAM is not intrinsically universal, in stark contrast to the simpler (single-tile addition only) abstract Tile Assembly Model which is intrinsically universal (\emph{The tile assembly model is intrinsically universal}, FOCS 2012). However, on the positive side, we prove that, for every fixed temperature~$\tau \geq 2$, temperature-$\tau$ 2HAM tile systems are indeed intrinsically universal: in other words, for each~$\tau$ there is a single universal 2HAM tile set $U$ that, when appropriately initialized, is capable of simulating the behavior of any temperature-$\tau$ 2HAM tile system. As a corollary of these results we find an infinite set of infinite hierarchies of 2HAM systems with strictly increasing simulation power within each hierarchy. Finally, we show that for each $\tau$, there is a temperature-$\tau$ 2HAM system that simultaneously simulates all temperature-$\tau$ 2HAM systems.
\end{abstract}

\input{intro}

\input{definitions}

\input{notIU_AltVersion}

\input{simulationsOverview}

\ifabstract
\fi
\iffull
\input{sqrtStrongSimulation}

\input{sqrtLogWeakSimulation}
\input{sim_all}

\fi
\bibliographystyle{abbrv} 
\bibliography{tam}

\appendix

\end{document}

%% file: intro.tex

\section{Introduction}
Self-assembly is the process through which unorganized, simple, components automatically coalesce according to simple local rules to form some kind of target structure. It sounds simple, but the end result can be extraordinary. For example, researchers have been able to self-assemble a wide variety of structures experimentally at the nanoscale, such as regular arrays~\cite{WinLiuWenSee98}, fractal structures~\cite{RoPaWi04,FujHarParWinMur07}, smiling faces~\cite{rothemund2006folding}, DNA tweezers~\cite{yurke2000dna}, logic circuits~\cite{qian2011scaling}, neural networks~\cite{qian2011neural}, and molecular robots\cite{DNARobotNature2010}. These examples are fundamental because they demonstrate that self-assembly can, in principle, be used to manufacture specialized geometrical, mechanical and computational objects at the nanoscale. Potential future applications of nanoscale self-assembly include the production of smaller, more efficient microprocessors and medical technologies that are capable of diagnosing and even treating disease at the cellular level.

Controlling nanoscale self-assembly for the purposes of manufacturing atomically precise components will require a bottom-up, hands-off strategy. In other words, the self-assembling units themselves will have to be ``programmed'' to direct themselves to do the right thing---efficiently and correctly. Thus, it is necessary to study the extent to which the process of self-assembly can be controlled in an algorithmic sense.

In 1998, Erik Winfree \cite{Winf98} introduced the abstract Tile Assembly Model (aTAM), an over-simplified discrete mathematical model of  nanoscale DNA  self-assembly pioneered by Seeman \cite{Seem82}. The aTAM essentially augments classical Wang tiling \cite{Wang61} with a mechanism for sequential ``growth'' of a tiling (in Wang tiling, only the existence of a valid, mismatch-free tiling is considered and not the order of tile placement). In the aTAM, the fundamental components are un-rotatable, but translatable square ``tile types'' whose sides are labeled with (alpha-numeric) glue ``colors'' and (integer) ``strengths''. Two tiles that are placed next to each other \emph{interact} if the glue colors on their abutting sides match, and they \emph{bind} if the strengths on their abutting sides match and sum to at least a certain (integer) ``temperature''. Self-assembly starts from a ``seed'' tile type and proceeds nondeterministically and asynchronously as tiles bind to the seed-containing-assembly. Despite its deliberate over-simplification, the aTAM is a computationally expressive model. For example, Winfree \cite{Winf98} proved that it is Turing universal, which implies that self-assembly can be directed by a computer program.

In this paper, we work in a generalization of the aTAM, called the \emph{two-handed}~\cite{Versus} (a.k.a., hierarchical \cite{CheDot12}, q-tile \cite{AGKS05g}, polyomino \cite{Luhrs08}) abstract Tile Assembly Model (2HAM). A central feature of the 2HAM is that, unlike the aTAM, it allows two ``supertile'' assemblies, each consisting of one or more tiles, to fuse together.   For two such assemblies to bind, they should not ``sterically hinder'' each other, and they should have a sufficient number of matching glues distributed along the interface where they meet. Hence the model includes notions of local  interactions (individual glues) and non-local interactions (large assemblies coming together).  In the 2HAM, an assembly of tiles is producible if it is either a single tile, or if it results from the stable combination of two other producible assemblies.

We study the \emph{intrinsic universality} in the 2HAM.  Intrinsic universality uses a special notion of simulation, where the simulator preserves the dynamics of the simulated system. For tile assembly systems this means that, modulo spatial rescaling, a simulator self-assembles the same assemblies as any simulated system, and even does this in the same way (via the same assembly sequences). In the field of cellular automata, the topic of intrinsic universality has given rise to a rich theory~\cite{ DurandRoka, bulkingI, Delorme-etal-2011,arrighi2012intrinsic, Goles-etal-2011, Ollinger-CSP08, ollingerRichard2011four} and indeed has also been studied in Wang tiling~\cite{LafitteW07,LafitteW08,LafitteW09} and  tile self-assembly~\cite{USA, IUSA, T1notIUforT2-2013,one-2012}. The aTAM has been shown to be intrinsically universal~\cite{IUSA}, meaning that there is a single set of tiles~$U$ that works at temperature 2, and when appropriately initialized, is capable of simulating the behavior of an arbitrary aTAM tile assembly system. Modulo rescaling, this single tile set $U$ represents the full power and expressivity of the entire aTAM model, at any temperature. On the other hand, it has been shown that at temperature-1, there is no tile set that can simulate the aTAM~\cite{T1notIUforT2-2013}. Interestingly, the latter negative result holds for 3D temperature-1 systems, which are known to be Turing universal~\cite{CooFuSch11}.  Here, we ask whether there is such a universal tile set for the 2HAM.

The theoretical power of non-local interaction in the 2HAM has been the subject of recent research. For example, Doty and Chen \cite{CheDot12} proved that, surprisingly, $N \times N$ squares do not self-assemble any faster in so-called {\em partial order} 2HAM systems than they do in the aTAM, despite being able to exploit massive parallelism.  More recently, Cannon, et al.~\cite{Versus}, while comparing the abilities of the 2HAM and the aTAM, proved three main results, which seem to suggest that the 2HAM is at least as powerful as the aTAM: (1) non-local binding in the 2HAM can dramatically reduce the tile complexity (i.e., minimum number of unique tile types required to self-assemble a shape) for certain classes of shapes; (2)  the 2HAM can simulate the aTAM in the following sense: for any aTAM tile system $\mathcal{T}$, there is a corresponding 2HAM tile system $\mathcal{S}$, which simulates the exact behavior---modulo connectivity---of $\mathcal{T}$, at scale factor 5; (3)  the problem of verifying whether a 2HAM system uniquely produces a given assembly is coNP-complete (for the aTAM this  problem is decidable in polynomial time \cite{ACGHKMR02}).

{\bf  Main results.\ } In this paper, we ask if the 2HAM is \emph{intrinsically universal}: does there exist a ``universal'' 2HAM tile set $U$ that, when appropriately initialized, is capable of simulating the behavior of an arbitrary 2HAM tile system? A positive answer would imply that such a tile set $U$ has the ability to model the capabilities of all 2HAM systems.\footnote{Note that the above simulation result of Cannon et al.\ does not imply that the 2HAM is intrinsically universal because (a) it is for 2HAM simulating aTAM, and (b) it is an example of a ``for all, there exists...'' statement, whereas intrinsic universality is a ``there exists, for all...'' statement.} Our first main result, Theorem~\ref{thm:2HAM-is-not-IU-general}, says that the 2HAM is \emph{not} intrinsically universal, which means that the 2HAM is  incapable of simulating itself. This statement stands in stark contrast to the case of the aTAM, which was recently shown to be intrinsically universal by Doty, Lutz, Patitz, Schweller, Summers and Woods \cite{IUSA}. Specifically, we show that for any sufficiently large temperature $\tau$, there is a temperature $\tau$ 2HAM system that cannot be simulated by any temperature $\tau' < \tau$ 2HAM system. It is worthy of note that, in order to prove this result,  we use a simple, yet novel combinatorial argument, which as far as we are aware of, is the first lower bound proof in the 2HAM that does not use an information-theoretic argument. In our proof of Theorem~\ref{thm:2HAM-is-not-IU-general} we show that the 2HAM cannot simulate massively cooperative binding, where the number of cooperative bindings is larger than the temperature of the simulator.

Our second main result, Theorem~\ref{sec:secondMainResult}, is positive:  we show, via constructions, that the 2HAM {\em is} intrinsically universal for fixed temperature, that is, the temperature~$\tau$ 2HAM can simulate the temperature~$\tau$ 2HAM. So although our impossibility result tells us that the 2HAM can not simulate ``too much'' cooperative binding, our positive result tells us it can indeed simulate {\em some} cooperative binding: an amount exactly equal to the temperature of the simulator.

As a corollary of these results, we get a separation between classes of 2HAM tile systems based on their temperatures. That is, we exhibit an infinite hierarchy of 2HAM systems, of strictly-increasing temperature, that cannot be simulated by lesser temperature systems but can downward simulate lower temperature systems. Moreover, we exhibit an infinite number of such hierarchies in Theorem~\ref{thm:infinite_hierarchy}.  
Thus, as was suggested as future work in~\cite{IUSA}, and as has been shown in the theory of cellular automata~\cite{Delorme-etal-2011}, we use the notion of intrinsic universality to classify, and separate, these tile assembly systems via their simulation ability.

As noted above, we show that temperature $\tau$ 2HAM systems are intrinsically universal. We actually show this for two different, seemingly natural, notions of simulation (called {\em simulation} and {\em strong simulation}), showing trade-offs between, and even within, these notions of simulation.  For both notions of simulation, we show tradeoffs between scale factor, number of tile types, and complexity of the initial configuration.  Finally, we show how to construct, for each $\tau$, a temperature-$\tau$ 2HAM system that simultaneously simulates all temperature-$\tau$ 2HAM systems. We finish with a conjecture:

\begin{conjecture}
There exists $c \in \mathbb{N}$, such that for each $\tau \geq c$, temperature $\tau$ 2HAM systems do not strongly simulate temperature $\tau -1 $ 2HAM systems.
\end{conjecture}

%% file: definitions.tex

\section{Definitions}\label{sec:definitions}
\subsection{Informal definition of the 2HAM}
The 2HAM \cite{AGKS05g,DDFIRSS07} is a generalization of the aTAM in that it allows for two assemblies, both possibly consisting of more than one tile, to attach to each other. Since we must allow that the assemblies might require translation before they can bind, we define a \emph{supertile} to be the set of all translations of a $\tau$-stable assembly, and speak of the attachment of supertiles to each other, modeling that the assemblies attach, if possible, after appropriate translation.
We now give a brief, informal, sketch of the 2HAM.

A \emph{tile type} is a unit square with four sides, each having a \emph{glue} consisting of a \emph{label} (a finite string) and \emph{strength} (a non-negative integer).   We assume a finite set $T$ of tile types, but an infinite number of copies of each tile type, each copy referred to as a \emph{tile}.
A \emph{supertile} is (the set of all translations of) a positioning of tiles on the integer lattice $\Z^2$.  Two adjacent tiles in a supertile \emph{interact} if the glues on their abutting sides are equal and have positive strength.
Each supertile induces a \emph{binding graph}, a grid graph whose vertices are tiles, with an edge between two tiles if they interact.
The supertile is \emph{$\tau$-stable} if every cut of its binding graph has strength at least $\tau$, where the weight of an edge is the strength of the glue it represents.
That is, the supertile is stable if at least energy $\tau$ is required to separate the supertile into two parts.
A 2HAM \emph{tile assembly system} (TAS) is a pair $\calT = (T,\tau)$, where $T$ is a finite tile set and $\tau$ is the \emph{temperature}, usually 1 or 2.
Given a TAS $\calT=(T,\tau)$, a supertile is \emph{producible}, written as $\alpha \in \prodasm{T}$ if either it is a single tile from $T$, or it is the $\tau$-stable result of translating two producible assemblies without overlap.
A supertile $\alpha$ is \emph{terminal}, written as $\alpha \in \termasm{T}$ if for every producible supertile $\beta$, $\alpha$ and $\beta$ cannot be $\tau$-stably attached.
A TAS is \emph{directed} if it has only one terminal, producible supertile.\footnote{We do not use this definition in this paper but have included it for the sake of completeness.}

\iffull
\subsection{Formal definition of the 2HAM}
\fi
\ifabstract
\later{
\section{Formal definition of the 2HAM}
}
\fi
\later{
\label{def:2ham_formal}
We now formally define the 2HAM.

Two assemblies $\alpha$ and $\beta$ are \emph{disjoint} if $\dom \alpha \cap \dom \beta = \emptyset.$
For two assemblies $\alpha$ and $\beta$, define the \emph{union} $\alpha \cup \beta$ to be the assembly defined for all $\vec{x}\in\Z^2$ by $(\alpha \cup \beta)(\vec{x}) = \alpha(\vec{x})$ if $\alpha(\vec{x})$ is defined, and $(\alpha \cup \beta)(\vec{x}) = \beta(\vec{x})$ otherwise. Say that this union is \emph{disjoint} if $\alpha$ and $\beta$ are disjoint.

The \emph{binding graph of} an assembly $\alpha$ is the grid graph
$G_\alpha = (V, E )$, where $V =
\dom{\alpha}$, and $\{\vec{m}, \vec{n}\} \in E$ if and only if (1)
$\vec{m} - \vec{n} \in U_2$, (2)
$\lab_{\alpha(\vec{m})}\left(\vec{n} - \vec{m}\right) =
\lab_{\alpha(\vec{n})}\left(\vec{m} - \vec{n}\right)$, and (3)
$\strength_{\alpha(\vec{m})}\left(\vec{n} -\vec{m}\right) > 0$.
Given $\tau \in \mathbb{N}$, an
assembly is $\tau$-\emph{stable} (or simply \emph{stable} if $\tau$ is understood from context), if it
cannot be broken up into smaller assemblies without breaking bonds
of total strength at least $\tau$; i.e., if every cut of $G_\alpha$
has weight at least $\tau$, where the weight of an edge is the strength of the glue it represents. In contrast to the model of Wang tiling, the nonnegativity of the strength function implies that glue mismatches between adjacent tiles do not prevent a tile from binding to an assembly, so long as sufficient binding strength is received from the (other) sides of the tile at which the glues match.

For assemblies $\alpha,\beta:\Z^2 \dashrightarrow T$ and $\vec{u} \in \Z^2$, we write $\alpha+\vec{u}$ to denote the assembly defined for all $\vec{x}\in\Z^2$ by $(\alpha+\vec{u})(\vec{x}) = \alpha(\vec{x}-\vec{u})$, and write $\alpha \simeq \beta$ if there exists $\vec{u}$ such that $\alpha + \vec{u} = \beta$; i.e., if $\alpha$ is a translation of $\beta$. Given two assemblies $\alpha,\beta:\Z^2 \dashrightarrow T$, we say $\alpha$ is a \emph{subassembly} of $\beta$, and we write $\alpha \sqsubseteq \beta$, if $S_\alpha \subseteq S_\beta$ and, for all points $p \in S_\alpha$, $\alpha(p) = \beta(p)$.
Define the \emph{supertile} of $\alpha$ to be the set $\ta = \setr{\beta}{\alpha \simeq \beta}$.
A supertile $\ta$ is \emph{$\tau$-stable} (or simply \emph{stable}) if all of the assemblies it contains are $\tau$-stable; equivalently, $\ta$ is stable if it contains a stable assembly, since translation preserves the property of stability. Note also that the notation $|\ta| \equiv |\alpha|$ is the size of the supertile (i.e., number of tiles in the supertile) is well-defined, since translation preserves cardinality (and note in particular that even though we define $\ta$ as a set, $|\ta|$ does not denote the cardinality of this set, which is always $\aleph_0$).

For two supertiles $\ta$ and $\tb$, and temperature $\tau\in\N$, define the \emph{combination} set $C^\tau_{\ta,\tb}$ to be the set of all supertiles $\tg$ such that there exist $\alpha \in \ta$ and $\beta \in \tb$ such that (1) $\alpha$ and $\beta$ are disjoint (steric protection), (2) $\gamma \equiv \alpha \cup \beta$ is $\tau$-stable, and (3) $\gamma \in \tg$. That is, $C^\tau_{\ta,\tb}$ is the set of all $\tau$-stable supertiles that can be obtained by ``attaching'' $\ta$ to $\tb$ stably, with $|C^\tau_{\ta,\tb}| > 1$ if there is more than one position at which $\beta$ could attach stably to $\alpha$.

It is common with seeded assembly to stipulate an infinite number of copies of each tile, but our definition allows for a finite number of tiles as well. Our definition also allows for the growth of infinite assemblies and finite assemblies to be captured by a single definition, similar to the definitions of \cite{jSSADST} for seeded assembly.

Given a set of tiles $T$, define a \emph{state} $S$ of $T$ to be a multiset of supertiles, or equivalently, $S$ is a function mapping supertiles of $T$ to $\N \cup \{\infty\}$, indicating the multiplicity of each supertile in the state. We therefore write $\ta \in S$ if and only if $S(\ta) > 0$.

A \emph{(two-handed) tile assembly system} (\emph{TAS}) is an ordered triple $\mathcal{T} = (T, S, \tau)$, where $T$ is a finite set of tile types, $S$ is the \emph{initial state}, and $\tau\in\N$ is the temperature. If not stated otherwise, assume that the initial state $S$ is defined $S(\ta) = \infty$ for all supertiles $\ta$ such that $|\ta|=1$, and $S(\tb) = 0$ for all other supertiles $\tb$. That is, $S$ is the state consisting of a countably infinite number of copies of each individual tile type from $T$, and no other supertiles. In such a case we write $\calT = (T,\tau)$ to indicate that $\calT$ uses the default initial state.  For notational convenience we sometimes describe $S$ as a set of supertiles, in which case we actually mean that  $S$ is a multiset of supertiles with infinite count of each supertile. We also assume that, in general, unless stated otherwise, the count for any supertile in the initial state is infinite.

Given a TAS $\calT=(T,S,\tau)$, define an \emph{assembly sequence} of $\calT$ to be a sequence of states $\vec{S} = (S_i \mid 0 \leq i < k)$ (where $k = \infty$ if $\vec{S}$ is an infinite assembly sequence), and $S_{i+1}$ is constrained based on $S_i$ in the following way: There exist supertiles $\ta,\tb,\tg$ such that (1) $\tg \in C^\tau_{\ta,\tb}$, (2) $S_{i+1}(\tg) = S_{i}(\tg) + 1$,\footnote{with the convention that $\infty = \infty + 1 = \infty - 1$} (3) if $\ta \neq \tb$, then $S_{i+1}(\ta) = S_{i}(\ta) - 1$, $S_{i+1}(\tb) = S_{i}(\tb) - 1$, otherwise if $\ta = \tb$, then $S_{i+1}(\ta) = S_{i}(\ta) - 2$, and (4) $S_{i+1}(\tilde{\omega}) = S_{i}(\tilde{\omega})$ for all $\tilde{\omega} \not\in \{\ta,\tb,\tg\}$.
That is, $S_{i+1}$ is obtained from $S_i$ by picking two supertiles from $S_i$ that can attach to each other, and attaching them, thereby decreasing the count of the two reactant supertiles and increasing the count of the product supertile. If $S_0 = S$, we say that $\vec{S}$ is \emph{nascent}.

Given an assembly sequence $\vec{S} = (S_i \mid 0 \leq i < k)$ of $\calT=(T,S,\tau)$ and a supertile $\tg \in S_i$ for some $i$, define the \emph{predecessors} of $\tg$ in $\vec{S}$ to be the multiset $\pred_{\vec{S}}(\tg) = \{\ta,\tb\}$ if $\ta,\tb \in S_{i-1}$ and $\ta$ and $\tb$ attached to create $\tg$ at step $i$ of the assembly sequence, and define $\pred_{\vec{S}}(\tg) = \{ \tg \}$ otherwise. Define the \emph{successor} of $\tg$ in $\vec{S}$ to be $\succ_{\vec{S}}(\tg)=\ta$ if $\tg$ is one of the predecessors of $\ta$ in $\vec{S}$, and define $\succ_{\vec{S}}(\tg)=\tg$ otherwise. A sequence of supertiles $\vec{\ta} = (\ta_i \mid 0 \leq i < k)$ is a \emph{supertile assembly sequence} of $\calT$ if there is an assembly sequence $\vec{S} = (S_i \mid 0 \leq i < k)$ of $\calT$ such that, for all $1 \leq i < k$, $\succ_{\vec{S}}(\ta_{i-1}) = \ta_i$, and $\vec{\ta}$ is \emph{nascent} if $\vec{S}$ is nascent.

The \emph{result} of a supertile assembly sequence $\vec{\ta}$ is the unique supertile $\res{\vec{\ta}}$ such that there exist an assembly $\alpha \in \res{\vec{\ta}}$ and, for each $0 \leq i < k$, assemblies $\alpha_i \in \ta_i$ such that $\dom{\alpha} = \bigcup_{0 \leq i < k}{\dom{\alpha_i}}$ and, for each $0 \leq i < k$, $\alpha_i \sqsubseteq \alpha$.  For all supertiles $\ta,\tb$, we write $\ta \to_\calT \tb$ (or $\ta \to \tb$ when $\calT$ is clear from context) to denote that there is a supertile assembly sequence $\vec{\ta} = ( \ta_i \mid 0 \leq i < k )$ such that $\ta_0 = \ta$ and $\res{\vec{\ta}} = \tb$. It can be shown using the techniques of \cite{Roth01} for seeded systems that for all two-handed tile assembly systems $\calT$ supplying an infinite number of each tile type, $\to_\calT$ is a transitive, reflexive relation on supertiles of $\calT$. We write $\ta \to_\calT^1 \tb$ ($\ta \to^1 \tb$) to denote an assembly sequence of length 1 from $\ta$ to $\tb$ and $\ta \to_\calT^{\leq 1} \tb$ ($\ta \to^{\leq 1} \tb$) to denote an assembly sequence of length 1 from $\ta$ to $\tb$ if $\ta \ne \tb$ and an assembly sequence of length 0 otherwise.

A supertile $\ta$ is \emph{producible}, and we write $\ta \in \prodasm{\calT}$, if it is the result of a nascent supertile assembly sequence. A supertile $\ta$ is \emph{terminal} if, for all producible supertiles $\tb$, $C^\tau_{\ta,\tb} = \emptyset$.\footnote{Note that a supertile $\ta$ could be non-terminal in the sense that there is a producible supertile $\tb$ such that $C^\tau_{\ta,\tb} \neq \emptyset$, yet it may not be possible to produce $\ta$ and $\tb$ simultaneously if some tile types are given finite initial counts, implying that $\ta$ cannot be ``grown'' despite being non-terminal. If the count of each tile type in the initial state is $\infty$, then all producible supertiles are producible from any state, and the concept of terminal becomes synonymous with ``not able to grow'', since it would always be possible to use the abundant supply of tiles to assemble $\tb$ alongside $\ta$ and then attach them.} Define $\termasm{\calT} \subseteq \prodasm{\calT}$ to be the set of terminal and producible supertiles of $\calT$. $\calT$ is \emph{directed} (a.k.a., \emph{deterministic}, \emph{confluent}) if $|\termasm{\calT}| = 1$.
} 

\subsection{Definitions for simulation}\label{sec:defsSim}
In this subsection, we formally define what it means for one 2HAM TAS to ``simulate'' another 2HAM TAS.  For a tileset $T$, let $A^T$ and $\tilde{A}^T$ denote the set of all assemblies over $T$ and all supertiles over $T$ respectively. Let $A^T_{< \infty}$ and $\tilde{A}^T_{< \infty}$ denote the set of all finite assemblies over $T$ and all finite supertiles over $T$ respectively.

In what follows, let $U$ be a tile set. An $m$-\emph{block assembly}, or {\em macrotile},  over tile set $U$ is a partial function $\gamma : \mathbb{Z}_m \times \mathbb{Z}_m \dashrightarrow U$, where $\mathbb{Z}_m = \{ 0,1,\ldots m-1 \}$.  Let $B^U_m$ be the set of all $m$-block assemblies over $U$. The $m$-block with no domain is said to be $\emph{empty}$.  For an arbitrary assembly $\alpha \in A^U$ define $\alpha^m_{x,y}$ to be the $m$-block defined by $\alpha^m_{x,y}(i,j) = \alpha(mx+i,my+j)$ for $0\leq i,j < m$.

For a partial function $R: B^{U}_m \dashrightarrow T$, define the \emph{assembly representation function} $R^*: A^{U} \dashrightarrow A^T$ such that $R^*(\alpha) = \beta$ if and only if $\beta(x,y) = R(\alpha^m_{x,y})$ for all $x,y \in \mathbb{Z}^2$.
    Further,
     $\alpha$ is said to map \emph{cleanly} to $\beta$ under $R^*$ if either (1) for all non empty blocks $\alpha^m_{x,y}$, $(x+u,y+v) \in \dom{\beta}$ for some $u,v \in \{-1,0,1\}$ such that $u^2+v^2 < 2$, or (2) $\alpha$ has at most one non-empty $m$-block $\alpha^m_{x,y}$. In other words, we allow for the existence of simulator ``fuzz'' directly north, south, east or west of a simulator  macrotile, but we exclude the possibility of diagonal fuzz.

For a given \emph{assembly representation function} $R^*$, define the \emph{supertile representation function} $\tilde{R}: \tilde{A}^{U} \dashrightarrow \mathcal{P}(A^T)$ such that $\tilde{R}(\ta) = \{R^*(\alpha) | \alpha \in \ta \}$. $\ta$ is said to \emph{map cleanly} to $\tilde{R}(\ta)$ if $\tilde{R}(\ta)\in \tilde{A}^T$ and $\alpha$ maps cleanly to $R^*(\alpha)$ for all~$\alpha \in \ta$.

In the following definitions, let $\mathcal{T} = \left(T,S,\tau\right)$  be a 2HAM TAS and, for some initial configuration $S_{\mathcal{T}}$, that depends on $\mathcal{T}$, let $\mathcal{U} = \left(U,S_{\mathcal{T}},\tau'\right)$ be a 2HAM TAS, and let $R$ be an $m$-block representation function $R: B^U_m \dashrightarrow T$.

\begin{definition}\label{scott-defn:alt-equiv-prod}
We say that $\mathcal{U}$ and $\mathcal{T}$ have \emph{equivalent productions} (at scale factor $m$), and we write $\mathcal{U} \Leftrightarrow_R \mathcal{T}$ if the following conditions hold:
\begin{enumerate}
    \item \label{scott-defn:simulate:equiv_prod_a}$\left\{\tilde{R}(\ta) | \ta \in \prodasm{\mathcal{U}}\right\} = \prodasm{\mathcal{T}}$.
    \item \label{scott-defn:simulate:equiv_prod_b}For all $\ta \in \prodasm{\mathcal{U}}$, $\ta$ maps cleanly to $\tilde{R}(\ta)$
\end{enumerate}
\end{definition}

\begin{definition}\label{scott-defn:alt-equiv-dynamic-t-to-s}
We say that $\mathcal{T}$ \emph{follows} $\mathcal{U}$ (at scale factor $m$), and we write $\mathcal{T} \dashv_R \mathcal{U}$ if, for any $\ta, \tb \in \prodasm{\mathcal{U}}$ such that $\ta \rightarrow_{\mathcal{U}}^1 \tb$, $\tilde{R}(\ta) \rightarrow_\mathcal{T}^{\leq 1} \tilde{R}\left(\tb\right)$.
\end{definition}

\begin{definition}\label{scott-defn:alt-equiv-dyanmic-s-to-t-weak}
We say that $\mathcal{U}$ \emph{weakly models} $\mathcal{T}$ (at scale factor $m$), and we write $\mathcal{U} \models^-_R \mathcal{T}$ if, for any $\ta, \tb \in \prodasm{\mathcal{T}}$ such that $\ta \rightarrow_\mathcal{T}^1 \tb$, for all $\ta' \in \prodasm{\mathcal{U}}$ such that $\tilde{R}(\ta')=\ta$, there exists an $\ta'' \in \prodasm{\mathcal{U}}$ such that $\tilde{R}(\ta'')=\tb$, $\ta' \rightarrow_{\mathcal{U}} \ta''$, and $\ta'' \rightarrow_{\mathcal{U}}^1 \tb'$ for some $\tb' \in \prodasm{\mathcal{U}}$ with $\tilde{R}\left(\tb'\right)=\tb$.
\end{definition}

\begin{definition}\label{scott-defn:alt-equiv-dyanmic-s-to-t-strong}
We say that $\mathcal{U}$ \emph{strongly models} $\mathcal{T}$ (at scale factor $m$), and we write $\mathcal{U} \models^+_R \mathcal{T}$ if for any $\ta$, $\tb \in \prodasm{\mathcal{T}}$ such that $\tilde{\gamma} \in C^{\tau}_{\ta , \tb}$, then for all $\ta', \tb' \in \prodasm{\mathcal{U}}$ such that $\tilde{R}(\ta')=\ta$ and $\tilde{R}\left(\tb'\right)=\tb$, it must be that there exist $\ta'', \tb'', \tilde{\gamma}' \in \prodasm{\mathcal{U}}$, such that $\ta' \rightarrow_{\mathcal{U}} \ta''$, $\tb' \rightarrow_{\mathcal{U}} \tb''$, $\tilde{R}(\ta'')=\ta$, $\tilde{R}\left(\tb''\right)=\tb$, $\tilde{R}(\tilde{\gamma}')=\tilde{\gamma}$, and $\tilde{\gamma}' \in C^{\tau'}_{\ta'', \tb''}$.
\end{definition}

\begin{definition}\label{scott-defn:alt-simulate}
Let $\mathcal{U} \Leftrightarrow_R \mathcal{T}$ and $\mathcal{T} \dashv_R \mathcal{U}$.
\begin{enumerate}
    \item \label{scott-defn:alt-weak-simulate} $\mathcal{U}$ \emph{simulates} $\mathcal{T}$ (at scale factor $m$) if $\mathcal{U} \models^-_R \mathcal{T}$.
    \item \label{scott-defn:alt-strong-simulate} $\mathcal{U}$ \emph{strongly simulates} $\mathcal{T}$ (at scale factor $m$) if $\mathcal{U} \models_R^+ \mathcal{T}$.
\end{enumerate}
\end{definition}

For simulation, we require that when a simulated supertile $\ta$ may grow, via one combination attachment, into a second supertile $\tb$, then any simulator supertile that maps to $\ta$ must also grow into a simulator supertile that maps to $\tb$. The converse should also be true.

For strong simulation, in addition to requiring that all supertiles mapping to $\ta$ must be capable of growing into a supertile mapping to $\tb$ when $\ta$ can grow into $\tb$ in the simulated system, we further require that this growth can take place by the attachment of $\emph{any}$ supertile mapping to $\tilde{\gamma}$, where $\tilde{\gamma}$ is the supertile that attaches to $\ta$ to get $\tb$.

\subsection{Intrinsic universality}
\newcommand{\REPL}{\mathsf{REPR}}
\newcommand{\frakC}{\mathfrak{C}}

Let $\REPL$ denote the set of all $m$-block (or macrotile) representation functions.
Let $\frakC$ be a class of tile assembly systems, and let $U$ be a tile set. 
We say $U$ is \emph{intrinsically universal} for $\frakC$ if there are computable functions $\mathcal{R}:\frakC \to \REPL$ and $\mathcal{S}:\frakC \to \left(A^U_{< \infty} \rightarrow \mathbb{N} \cup \{\infty\}\right)$, and a $\tau'\in\Z^+$ such that, for each $\mathcal{T} = (T,S,\tau) \in \frakC$, there is a constant $m\in\N$ such that, letting $R = \mathcal{R}(\mathcal{T})$, $S_\mathcal{T}=\mathcal{S}(\mathcal{T})$, and $\mathcal{U}_\mathcal{T} = (U,S_\mathcal{T},\tau')$, $\mathcal{U}_\mathcal{T}$ simulates $\mathcal{T}$ at scale $m$ and using macrotile representation function $R$.
That is, $\mathcal{R}(\mathcal{T})$ gives a representation function $R$ that interprets macrotiles (or $m$-blocks) of $\mathcal{U}_\mathcal{T}$ as assemblies of $\mathcal{T}$, and $\mathcal{S}(\mathcal{T})$ gives the initial state used to create the necessary macrotiles from $U$ to represent $\mathcal{T}$ subject to the constraint that no macrotile in $S_{\calT}$ can be larger than a single $m \times m$ square.

%% file: notIU_AltVersion.tex

\section{The 2HAM is not intrinsically universal}
In this section, we prove the main result of this paper: there is no universal 2HAM tile set 
that, when appropriately initialized,  is capable of simulating an arbitrary 2HAM system. That is, we prove that the 2HAM, unlike the aTAM, is not intrinsically universal.

\begin{theorem}\label{thm:2HAM-is-not-IU-general}
The 2HAM is not intrinsically universal.
\end{theorem}

We first prove Theorem~\ref{thm:2HAM-is-not-IU}, which says that, for any claimed  2HAM simulator $\mathcal{U}$, that runs at temperature $\tau' $, there exists a  2HAM system, with  temperature $\tau > \tau' $, that cannot be simulated by $\mathcal{U}$. We we use this as the main tool to prove Theorem~\ref{thm:no-2HAM-IU-tile-set}, a restatement of Theorem~\ref{thm:2HAM-is-not-IU-general}; our main result.

\begin{theorem}\label{thm:2HAM-is-not-IU}

Let $\tau \in \mathbb{N}, \tau \geq 2$.  For every tile set $U$, there exists a 2HAM TAS $\mathcal{T} = (T, S, \tau)$ such that for any initial configuration $S_{\mathcal{T}}$ over $U$ and $\tau' \leq \tau-1$, the 2HAM TAS $\mathcal{U} = \left(U,S_{\mathcal{T}},\tau'\right)$ does not simulate $\mathcal{T}$.

\end{theorem}

\textbf{The basic idea} of the proof of Theorem~\ref{thm:2HAM-is-not-IU} is to use Definitions~\ref{scott-defn:alt-equiv-dyanmic-s-to-t-weak} and~\ref{scott-defn:alt-equiv-prod} in order to exhibit two producible supertiles in $\mathcal{T}$, that do not combine in $\mathcal{T}$ because of a lack of total binding strength, and show that the supertiles that simulate them in $\mathcal{U}$ do combine in the (lower temperature) simulator $\mathcal{U}$. Then we argue that Definition~\ref{scott-defn:alt-equiv-dynamic-t-to-s} says that, because the simulating supertiles can combine in the simulator $\mathcal{U}$, then so too can the supertiles being simulated in the simulated system $\mathcal{T}$, which contradicts the fact that the two originally chosen supertiles from  $\mathcal{T}$ do not combine in $\mathcal{T}$.

\begin{proof}
Our proof is by contradiction. Therefore, suppose, for the sake of obtaining a contradiction, that there exists a universal tile set $U$ such that, for any 2HAM TAS $\mathcal{T} = (T,S,\tau)$, there exists an initial configuration $S_{\mathcal{T}}$ and $\tau'\leq \tau-1$, such that $\mathcal{U} = \left(U, S_{\mathcal{T}}, \tau'\right)$ simulates $\mathcal{T}$. Define $\mathcal{T} = (T, \tau)$ where $T$ is the tile set defined in Figure~\ref{fig:ladders-tile-set}, the default initial state is used, and $\tau > 1$.  Let $\mathcal{U} = \left(U, S_{\mathcal{T}}, \tau'\right)$ be the temperature $\tau'\leq \tau-1$ 2HAM system, which uses tile set $U$ and initial configuration $S_{\mathcal{T}}$ (depending on $\mathcal{T}$) to simulate $\calT$ at scale factor $m$. Let $\tilde{R}$ denote the assembly replacement function that testifies to the fact that $\mathcal{U}$ simulates $\mathcal{T}$.

\begin{figure}[htp]
\begin{center}
\includegraphics[width=4.5in]{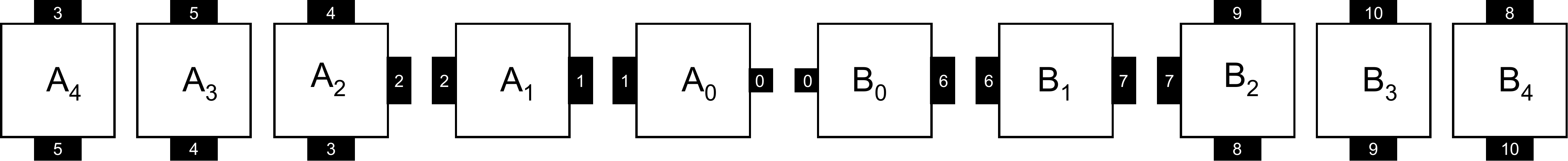}
\caption{The tile set for the proof of Theorem~\ref{thm:2HAM-is-not-IU}.  Black rectangles represent strength-$\tau$ glues (labeled 1-8), and black squares represent the strength-1 glue (labeled 0).}
\label{fig:ladders-tile-set}
\end{center}
\end{figure}

We say that a supertile $\tilde{l} \in \prodasm{\mathcal{T}}$ is a \emph{left half-ladder} of height $h \in \mathbb{N}$ if it contains $h$ tiles of the type A2 and $h-1$ tiles of type A3, arranged in a vertical column, plus $\tau$ tiles of each of the types A1 and A0. (An example of a left half-ladder is shown on the left in Figure~\ref{fig:ladders}. The dotted lines show positions at which tiles of type A1 and A0 could potentially attach, but since a half-ladder has exactly $\tau$ of each, only $\tau$ such locations have tiles.) Essentially, a left half-ladder consists of a single-tile-wide vertical column of height $2h-1$ with an A2 tile at the bottom and top, and those in between alternating between A3 and A2 tiles. To the east of exactly $\tau$ of the A2 tiles, an A1 tile is attached and to the east of each A1 tile, an A0 tile type is attached. These A1-A0 pairs, collectively, form the $\tau$ \emph{rungs} of the left half-ladder. We can define \emph{right} half-ladders similarly. A \emph{right half-ladder} of height $h$ is defined exactly the same way but using the tile types B3, B2, B1, and B0 and with rungs growing to the left of the vertical column. The east glue of A0 is a strength-$1$ glue matching the west glue of B0.

\begin{figure}[htp]
\begin{center}
\ifabstract
\includegraphics[width=1.8in]{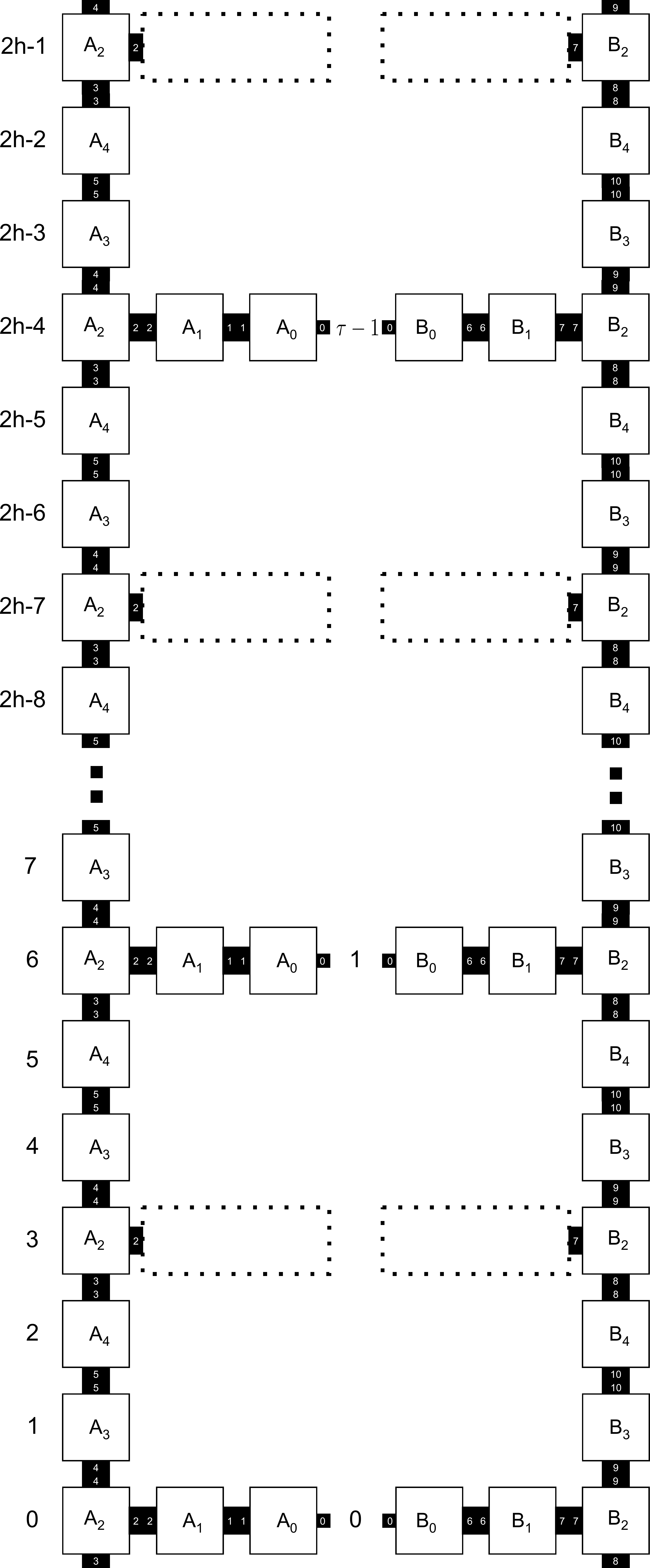}
\else
\includegraphics[width=2in]{images/ladders}
\fi
\caption{Example half-ladders with $\tau$ rungs.}
\label{fig:ladders}
\end{center}
\end{figure}

Let $LEFT \subseteq \prodasm{T}$ and $RIGHT \subseteq \prodasm{T}$ be the set of all left and right half-ladders of height $h$, respectively. Note that there are $\binom{h}{\tau}$ half-ladders of height $h$ in $LEFT$ ($RIGHT$). Define, for each $\tilde{l} \in LEFT$, the \emph{mirror image} of $\tilde{l}$ as the supertile $\bar{\tilde{l}} \in RIGHT$ such that $\bar{\tilde{l}}$ has rungs at the same positions as $\tilde{l}$.

For some $\tilde{l} \in LEFT$, we say that $\tilde{\hat{l}} \in \prodasm{U}$ is a \emph{simulator} left half-ladder of height $h$ if $\tilde{R}\left( \tilde{\hat{l}}\right) = \tilde{l}$. Note that $\tilde{\hat{l}}$ need not be unique. (One could even imagine $\tilde{\hat{l}}$ and $\tilde{\hat{l}}'$ satisfying $\tilde{R}\left(\tilde{\hat{l}}\right) = \tilde{l}$ and $\tilde{R}\left(\tilde{\hat{l}}'\right) = \tilde{l}$ but $\tilde{\hat{l}}$ and $\tilde{\hat{l}}'$ only differ by a single tile!) The notation $C^{\tau}_{\tilde{\alpha},\tilde{\beta}}$ is defined as the set of all supertiles that result in the $\tau$-stable combination of the supertiles $\tilde{\alpha}$ and $\tilde{\beta}$.

For some $\tilde{\hat{r}} \in \prodasm{U}$, we say that $\tilde{\hat{r}}$ is a \emph{mate} of $\tilde{\hat{l}}$ if $\tilde{R}\left(\tilde{\hat{r}}\right) = \tilde{r} \in RIGHT$, where $\tilde{r} = \bar{\tilde{l}}$, $C^{\tau}_{\tilde{l},\tilde{r}} \ne \emptyset$ (they combine in $\mathcal{T}$), and $C^{\tau-1}_{\hat{\tilde{l}},\hat{\tilde{r}}} \ne \emptyset$ (they combine in $\mathcal{U}$). For a simulator left half-ladder $\tilde{\hat{l}}$, we say that $\tilde{\hat{l}}$ is \emph{combinable} if $\tilde{\hat{l}}$ has a mate. Part~\ref{scott-defn:alt-weak-simulate} of Definition~\ref{scott-defn:alt-simulate} guarantees the existence of at least one combinable simulator left half-ladder for each left half-ladder. It is easy to see from Part~\ref{scott-defn:alt-weak-simulate} of Definition~\ref{scott-defn:alt-simulate} that an arbitrary simulator left half-ladder need not be combinable, since by Definition~\ref{scott-defn:alt-equiv-dyanmic-s-to-t-weak}, it may be a half-ladder $\tilde{\hat{l}} \in \prodasm{\mathcal{U}}$, which must first ``grow into'' a combinable left half-ladder $\tilde{\hat{l}}'$ (analogous to $\ta' \rightarrow_{\mathcal{U}} \ta''$ in Definition~\ref{scott-defn:alt-equiv-dyanmic-s-to-t-weak}).

Denote as $LEFT'$ some set that contains exactly one combinable simulator left half-ladder for each $\tilde{l} \in LEFT$. Note that, by Definitions~\ref{scott-defn:alt-equiv-prod} and~\ref{scott-defn:alt-equiv-dyanmic-s-to-t-weak}, there must be at least one combinable simulator left half-ladder $\tilde{\hat{l}}$ for each $\tilde{l}$, but that there also may be more than one, so the set $LEFT'$, while certainly not empty, need not be unique. By the definition of $LEFT'$, it is easy to see that $\left|LEFT'\right| = \binom{h}{\tau}$.  We know that each combinable simulator left half-ladder $\tilde{\hat{l}}$ has exactly $ \tau$ rungs, and furthermore, since glue strengths in the 2HAM cannot be fractional, it is the case that $\tau'$ of these rungs bind to (the corresponding rungs of) a mate with a combined total strength of at least $\tau'$.
(Note that some, but not all, of these $\tau'$ rungs may be redundant in the sense that they do not interact with positive strength.)

There are $\binom{h}{\tau'}$ ways to position/choose $\tau'$ rungs on a (simulator) half-ladder of height $h$. (Note that a rung on a simulator half-ladder need not be a $m\times m$ block of tiles but merely a collection of rung-like blocks that map to rungs in the input system $\mathcal{T}$ via $\tilde{R}$.)  Now consider the size $\binom{h}{\tau'}$ set of all possible rung positions, each denoted by a subset $X \subset \{0, 1, \ldots, h-1\}$, and the size $\binom{h}{\tau}$ set $LEFT'$.  For each simulated half-ladder $\tilde{\hat{l}} \in LEFT'$, there must exist a set of $\tau'$ rungs $X$ such that $\tilde{\hat{l}}$ binds to a mate via the rungs specified by $X$, with total strength at least $\tau'$.  As there are $\binom{h}{\tau}$ elements of $LEFT'$ and only $\binom{h}{\tau'}$ choices for $X$, the Generalized Pigeonhole Principle implies that there must be some set $LEFT'' \subset LEFT'$ with  $\left| LEFT'' \right| \geq {\binom{h}{\tau}}/{\binom{h}{\tau'}}$ such that every simulator left half-ladder in $LEFT''$ binds to a mate via the $\tau'
$ rungs specified by a single choice of $X$, with total strength at least $\tau'$.  In the case that $h \geq 2\tau$, we have that $\left| LEFT'' \right| \geq {\binom{h}{\tau}}/{\binom{h}{\tau'}} \geq {\binom{h}{\tau}}/{\binom{h}{\tau-1}} = \frac{h-\tau+1}{\tau}$.

Let $k = |U|^{4m^2}$, which is the number of ways to tile a neighborhood of four $m \times m$ squares from a set of $|U|$ distinct tile types. If $h = \tau\left(k^{\tau-1} + \tau\right)$, then $\left| LEFT'' \right| \geq k^{\tau-1} + 1$.
There are $k^{\tau'} \leq k^{\tau-1}$ ways to tile $\tau'$ neighborhoods that map to tiles of type A0  (plus any additional simulator fuzz that connects to simulated A0 tiles), under $\tilde{R}$, at the ends of the $\tau'$ rungs of a simulator left half-ladder.
This tells us that there are at least two (combinable) simulator left half-ladders $\tilde{\hat{l}}_1,\tilde{\hat{l}}_2 \in LEFT''$ such that $\tilde{\hat{l}}_1$ binds to a mate via the rungs specified by $X$, with total strength at least $\tau'$, $\tilde{\hat{l}}_2$ binds to a mate via the rungs specified by $X$, with total strength at least $\tau'$ and the rungs (along with any surrounding fuzz) specified by $X$ of $\tilde{\hat{l}}_1$ are tiled exactly the same as the rungs specified by $X$ of $\tilde{\hat{l}}_2$ are tiled. Thus, we can conclude that $\tilde{\hat{r}}$, a mate of $\tilde{\hat{l}}_1$, is a mate of $\tilde{\hat{l}}_2$. We can conclude this because, while $\tilde{\hat{l}}_1$ and $\tilde{\hat{l}}_2$ agree exactly along $\tau'$ of their rungs, they also each have one rung in a unique position and since consecutive rungs in $\mathcal{T}$ have at least two empty spaces between then, the offset simulator rungs (and even their fuzz) cannot prevent $\tilde{\hat{l}}_2$ from matching up with the mate of $\tilde{\hat{l}}_1$.

However, $\tilde{R}\left(\tilde{\hat{r}}\right) = \tilde{r} \in R$, $\tilde{R}\left(\tilde{\hat{l}}_2\right) = \tilde{l}_2 \in L$ but $C^{\tau}_{\tilde{r},\tilde{l}_2} = \emptyset$ because $\tilde{r}$ and $\tilde{l}_2$ differ from each other in one rung location and therefore interact in $\mathcal{T}$ with total strength at most $\tau-1$. This is a contradiction to Definition~\ref{scott-defn:alt-equiv-dynamic-t-to-s}, which implies $C^{\tau}_{\tilde{r},\tilde{l}_2} \ne \emptyset$.
 \end{proof}

We now have the main tool needed to prove our main Theorem~\ref{thm:2HAM-is-not-IU-general}, which we restate as follows.

\begin{theorem}\label{thm:no-2HAM-IU-tile-set}
There is no universal tile set $U$ for the 2HAM, i.e., there is no~$U$ such that, for all 2HAM tile assembly systems $\mathcal{T} = (T,S,\tau)$, there exists an initial configuration $S_{\mathcal{T}}$ over $U$ and temperature $\tau'$ such that $\mathcal{U} = \left(U,S_{\mathcal{T}},\tau'\right)$ simulates~$\mathcal{T}$.
\end{theorem}

\begin{proof}
Our proof is by contradiction, so assume that $U$ is a universal tile set. Denote as $g$ the strength of the strongest glue on any tile type in $U$. Let $\mathcal{T}' = (T',4g+1)$ be a modified version of the TAS $\mathcal{T} = (T,\tau)$ from the proof of Theorem~\ref{thm:2HAM-is-not-IU} with each $\tau$-strength glue in $T$ converted to a strength $4g+1$ glue in $T'$ (all other glues and labels are unmodified). By the proof of Theorem~\ref{thm:2HAM-is-not-IU}, we know that for any initial configuration $S_{\mathcal{T}}$ over $U$, $\mathcal{U} = \left(U,S_{\mathcal{T}},\tau'\right)$ does not simulate $\mathcal{T}$ for any $\tau' < 4g+1$. If $\tau' \geq 4g+1$, then the size of the largest supertile in $\mathcal{A}[\mathcal{U}]$ is 1 (since $g$ is the maximum glue strength in $U$, the supertiles in the initial state (input) $S_{\mathcal{T}}$ are not $\tau'$-stable and indeed no tile can bind to any assembly with strength $\geq 4g +1$), whence $U$ is not a universal tile set.%
\end{proof}

%% file: simulationsOverview.tex
\section{The temperature-$\tau$ 2HAM is intrinsically universal}\label{sec:simulationsOverview}

In this section we state our second main result, which states that for fixed temperature $\tau \geq 2$ the class of 2HAM systems at temperature $\tau$  is intrinsically universal. In other words,  for such $\tau$  there is a tile set that, when appropriately initialized, simulates any temperature $\tau$ 2HAM system. Denote as $\textrm{2HAM}(k)$ the set of all 2HAM systems at temperature $k$.
\begin{theorem}\label{sec:secondMainResult}
  For all $ \tau \geq 2  $, 2HAM($\tau$) is intrinsically universal.
\end{theorem}
\ifabstract
In the full version of this paper we prove this theorem for two different, but seemingly natural notions of simulation.
\else
We prove this theorem for two different, but seemingly natural notions of simulation.
\fi
The first, simply called {\em simulation}, is where we require that when a simulated supertile $\ta$ may grow, via one 
attachment, into a second supertile $\tb$, then any simulator supertile that maps to $\ta$ must also grow into a simulator supertile that maps to $\tb$. The converse should also be true.
\ifabstract
\else
Results for {\em simulation}  are given in Section~\ref{sec:weakSimulation}.

\fi
The second notion, called {\em strong simulation}, is a stricter definition where  in addition to requiring that all supertiles mapping to $\ta$ must be capable of growing into a supertile mapping to $\tb$ when $\ta$ can grow into $\tb$ in the simulated system, we further require that this growth can take place by the attachment of $\emph{any}$ supertile mapping to $\tilde{\gamma}$, where $\tilde{\gamma}$ is the supertile that attaches to $\ta$ to get~$\tb$. \ifabstract
\else
Results for {\em strong simulation}  are given in Section~\ref{sec:strongSimulation}.

\fi
For each of the two notions of simulation we provide three results, and in all cases we provide lower scale factor for simulation relative to strong simulation.  \ifabstract  \else
Specifically, strong simulation achieves a modest linear scale factor simulation, but a compact single input assembly is sufficient to encode the entire simulated system.  In contrast, for simulation (i.e.\ not strong), we are able to achieve a logarithmic scale factor in the size of the simulated system.  However, such small scale requires that the simulated system be encoded in a larger (linear) number of input assemblies, in order to  describe the simulated system without loss of information.
\fi

\ifabstract
\else
For strong simulation, Theorems~\ref{thm:strongSim2},~\ref{thm:strongSim1}, and~\ref{thm:strongSim3} provide three different proofs of our main positive result (Theorem~\ref{sec:secondMainResult}) with each of the three providing different trade-offs between number tile types, scale factor, and complexity of initial configuration for the simulator. For simulation,
Theorems~\ref{thm:weak1},~\ref{thm:weak2} and~\ref{thm:weak3}, provide similar trade-offs.
\fi

When we combine our negative and positive results, we get a separation between classes of 2HAM tile systems based on their temperatures.

\begin{theorem}\label{thm:infinite_hierarchy}
There exists an infinite number of infinite hierarchies of 2HAM systems with strictly-increasing power (and temperature) that can simulate downward within their own hierarchy.
\end{theorem}
\begin{proof}
Our first main result (Theorem~\ref{thm:2HAM-is-not-IU}) tells us that the temperature-$\tau$ 2HAM cannot be simulated by any temperature $\tau'  < \tau$ 2HAM.
Hence we have, for all $i > 0, c \geq 4$, $\textrm{2HAM}\left(c^i\right) \succ \textrm{2HAM}\left(c^{i - 1}\right)$, where $\succ$ is the relation ``cannot be simulated by''. Moreover, Theorem~\ref{sec:secondMainResult} tells us that temperature $\tau$ 2HAM is intrinsically universal for fixed temperature $\tau$. Suppose that $\tau' < \tau$ such that $ \tau / \tau'   \in \mathbb{N}$. Then the temperature-$\tau$ 2HAM can simulate temperature $\tau'$ (by simulating strength $g \leq \tau'$ attachments in the temperature $\tau'$ system with strength $g {\tau}/{\tau'}$ attachments in the temperature $\tau$ system). Thus, for all $0 < i' \leq i$, $\textrm{2HAM}\left(c^{i}\right)$ can simulate,  via Theorem~\ref{sec:secondMainResult}, $\textrm{2HAM}\left(c^{i'}\right)$. The theorem follows by noting that our choice of $c$ was arbitrary.
\end{proof}

We have shown that for each $\tau \geq 2$ there exists a single set of tile types~$U_\tau$, and a set of input supertiles over $U_\tau$, such that the 2HAM system strongly simulates any 2HAM TAS $\mathcal{T}$. A related question is:  does there exist a tile set that can simulate, or strongly simulate, all temperature $\tau$ 2HAM TASs simultaneously?  Surprisingly, the answer is yes!
\ifabstract
\else
The proof of the following theorem is given in Section~\ref{sec:simAll}.
\fi

\begin{theorem}\label{thm:ham-for-all-overview-version}
For each $\tau > 1$, there exists a 2HAM system $\mathcal{S} = (U_{\tau},\tau)$ which simultaneously strongly simulates all 2HAM systems $\mathcal{T} = (T,\tau)$.
\end{theorem}

%% file: sqrtStrongSimulation.tex
\section{The temperature-$\tau$ 2HAM is intrinsically universal: strong simulation}\label{sec:strongSimulation}

We present a total of six simulation results, three in this section and three in Section~\ref{sec:weakSimulation}.  In this section the three simulation results exhibit for any integer $\tau \geq 2$, a single set of tiles $U_\tau$ that at temperature~$\tau$ {\em strongly simulates} any temperature~$\tau$ 2HAM system, given a proper configuration of initial assemblies over~$U_\tau$.
Each of the three results in this section depict different tradeoffs between the number of encoded input supertiles and the scale of the simulation.
 Theorem~\ref{thm:strongSim3} achieves this while having the  simulated tile set encoded as a single input assembly.

For the aTAM it is known~\cite{IUSA}  that  there is a single tile set $U$ that simulates any aTAM tile assembly system $\mathcal{T}$, when initialized with a {\em single} seed assembly~$\sigma_{\mathcal{T}}$ that encodes  $\mathcal{T}$.
Assembly proceeds by  additions of single tiles to this seed. In this paper, where we study the 2HAM, it makes sense to allow the simulator to be programmed with multiple copies of the seed (input), rather than a single copy. In particular, this is the case in Theorem~\ref{thm:strongSim3} for strong simulation (and thus, also the weaker notion of {\em simulation}) where the simulator's input consists of infinitely many copies of both a single seed supertile, as well as the simulator's tiles.  However, the definition of input configuration allows fancier input configurations: it  permits us to have numerous distinct  seed assemblies. By exploiting this we achieve better scaling  in Theorems~\ref{thm:strongSim2}, \ref{thm:strongSim1},~\ref{thm:weak1}, \ref{thm:weak2}, and~\ref{thm:weak3} than in the single-seed case of Theorem~\ref{thm:strongSim3} (the six results  also give trade-offs in numbers of tile types in the simulator).
However, since these improvements in scaling (and possibly number of tile types) come at the expense of having many seed assemblies in the simulator, there is an intuitive sense in which ``less'' self-assembly, or at least a different form of self-assembly, is happening as fewer, and larger, assemblies themselves can act as large polyomino jigsaw pieces that come together to simulate tiles.  It is worth pointing out that our main result, an impossibility result, holds {\em despite} the fact that the simulator may try to use a large number of complicated-looking input assemblies.

Let $|| M ||$ denote the number of distinct elements in the multiset $M$, i.e.\ $|| M || = |M'|$ is the cardinality of the set $M'$ defined by ignoring multiplicities in the multiset $M$.    Let $T_{\mathrm{sup}}$ denote the set of supertiles induced by a tile set $T$.
By this, we simply mean $T_{\mathrm{sup}}$ is the set of supertiles formed by taking all tiles in $T$ and translating them  to all locations in $\mathrm{Z}^2$.

\subsection{Strong simulation with small scale and few tile types}

\begin{theorem}\label{thm:strongSim2}
For every $\tau \geq 2$
there exists a single set of tile types $U_\tau$, with $|U_\tau| = O(1)$ (i.e.\ independent of $\tau$),
such that for all 2HAM systems $\mathcal{T} = (T,\sigma,\tau)$,
there is a set $I_\mathcal{T}$ of $||\sigma||$ input supertiles
such that the 2HAM system $\mathcal{U}_{\mathcal{T}} = (U_\tau, U_{\tau, \mathrm{sup}} \cup I_\mathcal{T},  \tau)$ strongly simulates $\mathcal{T}$ at scale $O( \sqrt{|G| (\tau + \log |G|) } )$, where $G$ is the set of glues in $T$.
\end{theorem}

\begin{figure}[t] 
\begin{center}
  \subfloat[][]{%
        \label{fig:strongSimSketch}%
        \centering
        \includegraphics[width=0.5\linewidth]{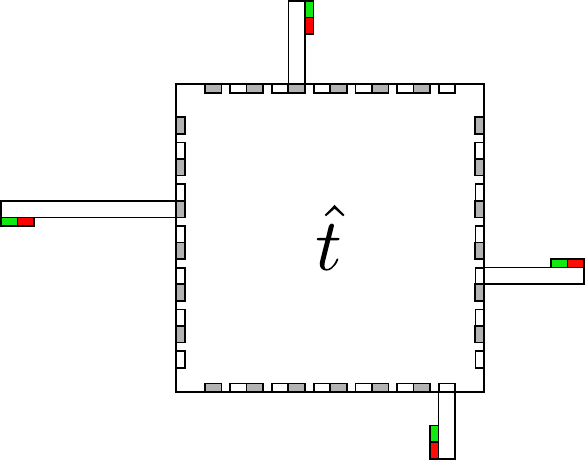}
       \hspace{0.05\linewidth}
       }%
  \subfloat[][]{%
        \label{fig:strongSimArmsDetail}%
        \centering
        \includegraphics[width=0.3\linewidth]{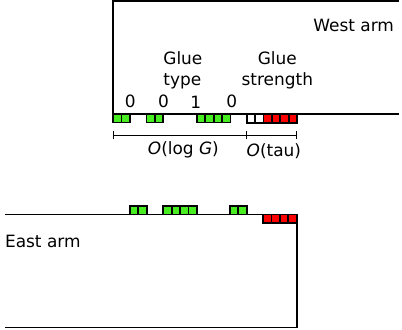}
        }
\caption{Strong simulation. (a) Input assembly design. For each tile type $t$ in the simulated TAS there is a unique macrotile $\hat{t}$. Glue pads and binding pads are shown in green and red, respectively. North arms project from grey regions so that they are staggered relative to south arms, west and east arms act similarly. For each unique glue $g$ in the simulated TAS, there is a unique pair of coordinates for the glue/binding pads, one denoting the position along the side of the macrotile (i.e. the particular grey/white location), and one denoting the distance that the arm projects outward from the macrotile.  (b) Detail showing two arms with matching glues (encoded as binary bumps and dents using green tiles) of strength 4 (encoded using 4 red tiles), from a temperature $\tau = 6$ system. Each red tile provides unit binding strength, and so macrotiles bind with strength equal to the number of matching red tiles. Only red tiles provide non-zero binding strength.  }
\end{center}
\end{figure}

\paragraph{Construction.} For a given simulated TAS $\mathcal{T} = (T,\sigma,\tau)$, the TAS $\mathcal{U}_\mathcal{T}$ represents the initial state of $\mathcal{T}$ as follows.

Each {\em singleton} tile type $t \in T$ that is used in the initial state $\sigma$ is represented as a macrotile of the form shown in Figure~\ref{fig:strongSimSketch} (larger input supertiles from $\sigma$ are described below). Tile type $t$ is mapped to a unique macrotile $\hat{t}$ as follows. First, each glue $g \in G$, from the tile set $T$, is uniquely encoded as a pair $\hat{g} = (i,j) \in \ell \times \ell$ where $\ell=\{x \in \mathbb{N}|x < \left\lceil \sqrt{|G|} \right\rceil\}$ (using the inverse of some simple pairing function). Assume that tile type $t$ has the four glues $(g_{\mathrm{north}}, g_{\mathrm{west}}, g_{\mathrm{south}}, g_{\mathrm{east}})$. Next, these glues are encoded as the four coordinates $(\hat{g}_{\mathrm{north}}, \hat{g}_{\mathrm{west}}, \hat{g}_{\mathrm{south}}, \hat{g}_{\mathrm{east}})$ using the above encoding. The macrotile $\hat{t}$ is composed of five parts: a square {\em body} (of size $k \times k$, where $ k= O( \sqrt{G (\tau + \log G) } )$), and four arms so that one arm is placed in each of four square regions, each of size $k\times k$, and adjacent to the body. The location of the red and green glue-binding pad on an arm uniquely encodes the relevant glue of tile type $t$: for example $g_{\mathrm{west}}$ is encoded by the glue-binding pad location $( k i /\lceil \sqrt{|G|} \rceil, k j /\lceil \sqrt{|G|}\rceil)$, where $\hat{g}_{\mathrm{west}} =  (i,j)$. (The arm length  is $k i /\lceil \sqrt{|G|} \rceil$ and its position along the supertile side is $k j /\lceil \sqrt{|G|}\rceil$.)   Arm lengths on east sides are ``complimentary'' to those on west sides, in the following sense. Let $g = g_{\mathrm{west}} = g_{\mathrm{east}}$ be some glue that appears on the east side of one tile and west side of another. As described above, the arm with the pad $\hat{g}_{\mathrm{west}}$ has length such that the glue pad appears at location $k i /\lceil \sqrt{|G|} \rceil$. However, the arm with the pad $\hat{g}_{\mathrm{east}}$ is defined to have length such that the glue pad appears at location $k - (k i /\lceil \sqrt{|G|} \rceil)$. The same complementarity trick is used for north and south arms.  Finally, as can be seen in Figure~\ref{fig:strongSimSketch}, arms are {\em staggered} relative to each other: north and west arms sit on grey patches, south and east arms sit on white patches.

Figure~\ref{fig:strongSimArmsDetail} shows the individual glue-binding pads at the end of two arms. Each glue $g \in G$ is uniquely represented as a bit sequence, which in turn is represented using bumps and dents in a $1 \times 4 \left\lceil \log |G| \right\rceil$ region shown in green on the west arm in Figure~\ref{fig:strongSimArmsDetail}. The same bump-dent pattern would be used on a north arm. For east and south arms, the complementary bump-dent pattern is used. To illustrate this, Figure~\ref{fig:strongSimArmsDetail} shows a west and east arm that share the same glue type. It can be seen that the arms are able to be translated so that the green regions {\em fit} together.  A glue-binding pad also encodes the binding strength $\mathrm{str}(g) \leq \tau$ of its represented glue $g$, in a straightforward way as a sequence of $\tau - \mathrm{str}(g)$ white tiles followed by $ \mathrm{str}(g)$ red tiles.  Each red tile, on the west arm, exposes a single strength 1 glue to its south. The red tiles are the {\em only} tiles on the entire macrotile $\hat{t}$ that expose positive strength glues. A matching east arm, as shown in Figure~\ref{fig:strongSimArmsDetail}, exposes the same number $ \mathrm{str}(g)$ of matching strength 1 glues from its $ \mathrm{str}(g)$ red tiles.\footnote{Note that in the glue-binding pad region there are no ``single tile'' bumps: this ensures that the simulator tile set $U$ does not contain strength $\tau$ glues, which in turn simplifies our construction.}

Due to their complementary green regions, and their matching sequence of exactly $\mathrm{str}(g)$ red tiles, the two arms shown in Figure~\ref{fig:strongSimArmsDetail} can be translated so that they bind together with strength $\mathrm{str}(g)$.

This completes the description of the encoding of the {\em singleton} tile types $t \in T$ that are used in the initial state $\sigma$. The remaining  supertiles of $\sigma$ (i.e. of size $> 1$) are encoded as in the following paragraph.

Consider 2 macrotiles $\hat{t}_1,\hat{t}_2$, that represent tiles $t_1, t_2$. From the above description, it can be seen that $\hat{t}_1$ and $\hat{t}_2$ can be positioned so that their bodies' centers lie on the same horizontal line, at exactly $2k$ distance apart, such that $\hat{t}_1,\hat{t}_2$ do not intersect. Furthermore, if $t_1$ and $t_2$ have a matching glue on their east and west sides, respectively, then the glue-binding pads of the east arm of~$\hat{t}_1$ and the west arm of~$\hat{t}_2$ will be positioned so that their matching bump-dent patterns  interlock, and their arms bind with whatever strength $t_1$ and $t_2$ bind. Finally,  if $t_1$ and $t_2$ do not have matching glues on their east and west sides, respectively, then it is the case the glue-binding pads of the two arms do not touch, and indeed their ``mismatching arms'' do not intersect. This holds for  the only other potential binding position (i.e.\ north-south)  of two arbitrary tiles $t_1$ and $t_2$.  Figure~\ref{fig:strongSimOneSeedOverview} shows six macrotiles translated into position to simulate an assembly of 6 tiles in some simulated TAS. The figure shows 5 matching arms (simulating matching glues) and two mismatched arms (simulating mismatched glues). Supertiles of size $\> 1$ in $\sigma$ are encoded in this manner.

We have now completely specified the initial state from which the self-assembly process proceeds in  $\mathcal{U}_\mathcal{T}$.

\begin{figure}[t]
\begin{center}
        \includegraphics[width=1\linewidth]{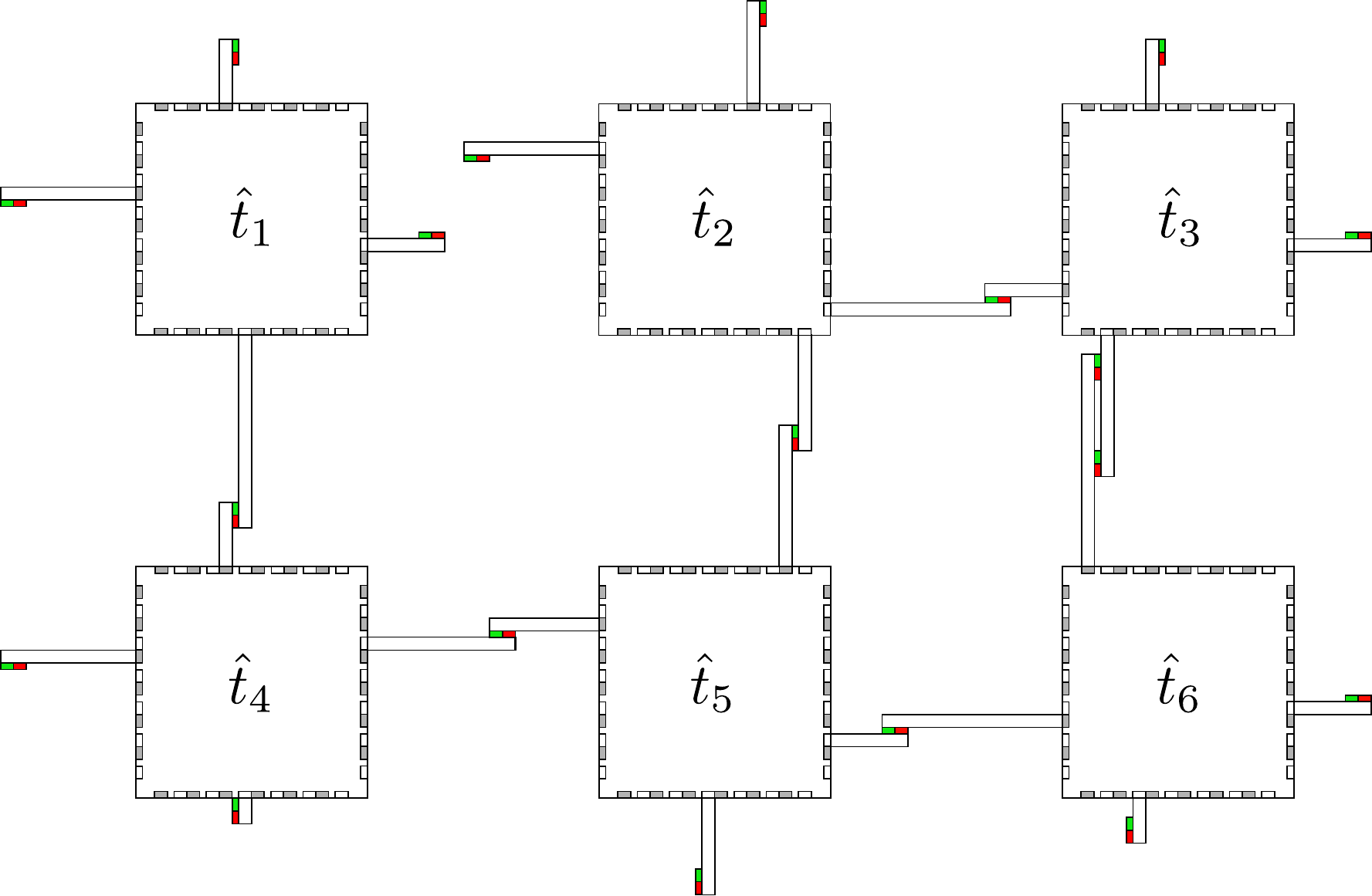}
\caption{Strong simulation. Example assembly with 6 macrotiles and 2 simulated mismatches. The square ``body'' of each tile is of size $k \times k$, and the assembly sits on a $2k \times 2k$ lattice.}
\label{fig:strongSimAssembly}
\end{center}
\end{figure}

\paragraph{Scale.}
Each glue $g \in G$ of the simulated TAS $\mathcal{T}$ is encoded using a glue-binding pad consisting of $O(\log G)$ (green) tiles and $\tau$ (red) tiles. There are $|G|$ such glue-binding pads, which are rasterized into each of four $k \times k$ regions, where  $k = O( \sqrt{G (\tau + \log G) } )$. To see that the tile set size is a constant (i.e.\ independent of the simulated TAS $\mathcal{T}$), note that all tiles on the outside of the macrotile expose strength 0 glues, except for the red tiles, which each expose the same strength 1 glue. Hence, the interior of each macrotile can be filled in using a single filler tile, with a constant size set of tile types used to fill the exterior.

\paragraph{Correctness of Simulation.}  Via the following two cases, macrotiles stay ``on-grid.'' (1) Due to their arm lengths, if two {\em matching} green glue pads bind (causing the binding of two macrotiles) then the combined (horizontal or vertical) arm length is exactly $k$. Thus {\em matching} macrotiles only bind in a way that their centers are exactly  distance   $2k$ apart. (2) Due to the green glue pad design, if two glue-binding pads mismatch (i.e.\ represent two mismatching glues in $T$) then they can not bind, since the two green pads sterically hinder each other. Taken together, this means that whenever two macrotiles (that encode two tile types $t_1, t_2 \in T$) bind, they are always positioned on a  $2k \times 2k$ square grid.

This immediately implies that whenever larger assemblies, with multiple macrotiles, bind, they have all of their tiles positioned on a  $2k \times 2k$ square grid.

It remains to show that, due to our macrotile design,  $\mathcal{U}_\mathcal{T}$  strongly simulates~$\mathcal{T}$ (i.e.\ point 2 of Definition~\ref{scott-defn:alt-simulate}). Firstly, due the glue-binding pad design and the fact that $\mathcal{T}$ and $\mathcal{U}_\mathcal{T}$ work at the same temperature $\tau$, a pair of supertiles in $\mathcal{T}$ bind if and only if their corresponding encoded pair of supertiles bind in  $\mathcal{U}_\mathcal{T}$.  This, taken together with the fact that the initial state of $\mathcal{U}_\mathcal{T}$ is an encoding of the initial state of $\mathcal{T}$, implies that the two systems have the same dynamics (in the strong sense), thus satisfying Definitions~\ref{scott-defn:alt-equiv-dynamic-t-to-s} and \ref{scott-defn:alt-equiv-dyanmic-s-to-t-strong}. Definition~\ref{scott-defn:alt-equiv-prod}(1) (equivalent production) is  satisfied since equivalent dynamics implies equivalent production, and Definition~\ref{scott-defn:alt-equiv-prod}(2) is satisfied as our choice of macrotile design directly implies that  each supertile produced in $\mathcal{U}_\mathcal{T}$ maps cleanly to a supertile in $\mathcal{T}$. Taken together, these facts are sufficient to satisfy Definition~\ref{scott-defn:alt-simulate}(2).

This completes the proof of Theorem~\ref{thm:strongSim2}.

\subsection{Strong simulation with smaller scale  but more tile types}

\begin{theorem}\label{thm:strongSim1}
For each integer $\tau \geq 2$ there exists a single set of tile types $U_{\tau}$, $|U_{\tau}| = O(\tau)$, such that for any 2HAM system $\mathcal{T} = (T,\sigma,\tau)$, there exists a set $I_\mathcal{T}$ of $||\sigma||$ input supertiles   such that the 2HAM system $\mathcal{U}_{\mathcal{T}} = (U_\tau, U_{\tau,\mathrm{sup}} \cup I_\mathcal{T},  \tau)$ strongly simulates $\mathcal{T}$ at scale $O( \sqrt{|G|  \log |G| } )$, where $G$ is the set of glues in $T$.
\end{theorem}

Theorem~\ref{thm:strongSim1} simply is a trade-off in tile types for scale factor with Theorem~\ref{thm:strongSim2}. The construction for Theorem~\ref{thm:strongSim2} is modified so that where before each arm had $\mathrm{str}(g)$ red tiles each exposing a strength 1 glue, now each arm has a {\em single} red tile that exposes one strength $\mathrm{str}(g)$ glue. This adds an additional $\tau -1$ tile types to the previous construction: one red tile type for each strength $s$ where $\{ 1 \leq s \leq \tau \} $). However, by shrinking the red binding pads to size~1, it results in a reduction in scale factor to $O( \sqrt{|G|  \log |G| } )$, giving the  statement of Theorem~\ref{thm:strongSim1}.

\begin{figure}[t] 
\begin{center}
  \subfloat[][]{%
        \label{fig:strongSimOneSeedOverview_seed}%
        \centering
        \includegraphics[width=0.4\linewidth]{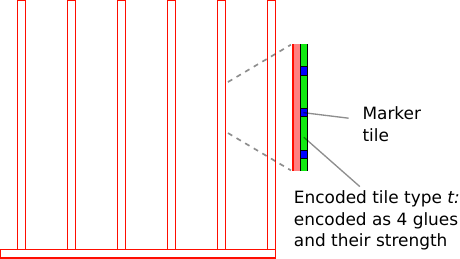}
       \hspace{0.0\linewidth}  
       }%
  \subfloat[][]{%
        \label{fig:strongSimOneSeedOverview}%
        \centering
        \includegraphics[width=0.6\linewidth]{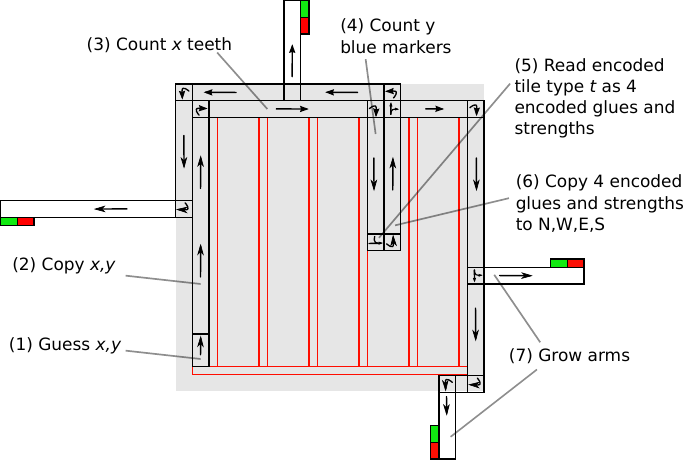}
        }
\caption{Overview of simulation for Theorem~\ref{thm:strongSim3}. (a) Input supertile. The red comb structure encodes the entire input simulated tile set $T_{\mathrm{in}}$. Each of the $\lceil\sqrt{|T_{\mathrm{in}}|}\rceil$ teeth of the comb encodes $\lceil\sqrt{|T_{\mathrm{in}}|}\rceil$ tile types from $T_{\mathrm{in}}$. Each simulated tile $t \in T_{\mathrm{in}}$ is encoded using $O(\log |G| + \log \tau)$  tiles (shown in green in zoom-in) that encode the 4 glues of $t$, and their strengths, in binary. Encoded tile types are separated by a single marker tile shown in blue.  (b) Growth of an input supertile into a macrotile that represents some tile type $t \in T$ from the simulated TAS $\mathcal{T}$.  Numbers and arrows are used to indicate order and direction of growth. Growth of a crawler (black outline) begins in the lower left. (1) The crawler guesses two integers $x,y$, where $0\leq x,y \leq \lfloor\sqrt{|T_{\mathrm{in}}|}\rfloor$, and (2) copies these values to the top left corner. (3) The crawler grows right, counting $x$ teeth, (4) grows south counting $y$ blue tiles (see (a)), and (5) copies the encoded tile type $t$ (as 4 encoded glues with strengths). (6) The glue information is copied to the four sides of the macrotile, and the glues are used as input to counters that (7) grow arms to the relevant lengths, including the red and green binding pads.    }
\label{fig:strongSimOneSeed}
\end{center}
\end{figure}

\subsection{Strong simulation with larger scale  but with fewer tile types and with only a  single macrotile}\label{subsec:strongSim3Sec}

The following Theorem provides an interesting tradeoff with the prior two results. It uses a  larger scale factor than Theorems~\ref{thm:strongSim2} and~\ref{thm:strongSim1}, but gives a universal 2HAM TAS for each $\tau$ where the {\em simulated tile set} is encoded in a single simulator supertile. This supertile has the ability to ``grow into'' any encoded tile in the simulated system $\mathcal{T}$.

\begin{theorem}\label{thm:strongSim3}
For every $\tau \geq 2$
there exists a single set of tile types $U_\tau$, with $|U_\tau| = O(1)$ (i.e.\ independent of $\tau$),
such that for all 2HAM systems $\mathcal{T} = (T,\sigma,\tau)$,
there is a single supertile~$s_T$ (that represents $T_{\mathrm{sup}} \cap \sigma$) and a set $I_\sigma$ containing $||\sigma - T_{\mathrm{sup}}||$ supertiles (that represent   $\sigma - T_{\mathrm{sup}}$)  such that the 2HAM system $\mathcal{U}_{\mathcal{T}} = (U_\tau, U_{\tau.\mathrm{sup}} \cup \{s_T\} \cup I_\sigma,  \tau)$ strongly simulates $\mathcal{T}$ at scale $O( \sqrt{|T|}(\log \tau + \log |G|) + \sqrt{ |G| (\tau + \log |G|) } )$, where $G$ is the set of glues in $T$.
\end{theorem}

\paragraph{Construction.} For notation, let $T_{\mathrm{in}}$ be the set of  singleton tiles used to make the multiset of singleton supertiles $T'  = T_\mathrm{sup} \cap \sigma$. The 2HAM simulator $\mathcal{U}_\mathcal{T}$ starts simulation from (A) a single input supertile $s_T$ that encodes {\em all} singleton tiles from $T_{\mathrm{in}} $, (B) $||\sigma   - T_{\mathrm{sup}}||$  supertiles that encode non-singleton input supertiles (i.e.\ they encode $\sigma- T_{\mathrm{sup}} $), as well as (C)  its own tile set $U_\tau$.

We begin by describing  (A), the single supertile $s_T$.  Figure~\ref{fig:strongSimOneSeedOverview_seed} shows the supertile assembly $s_T$, which we call a ``comb.''   The comb  encodes the entire simulated set of tile types, $T_{\mathrm{in}}$, that are actually used by $\mathcal{T}$: each of the $\lceil\sqrt{|T_{\mathrm{in}}|}\rceil$ teeth of the comb encodes $\leq \lceil\sqrt{|T_{\mathrm{in}}|}\rceil$ tile types from~$T_{\mathrm{in}}$. Each simulated tile $t \in T_{\mathrm{in}}$ is encoded using $O(\log |G| + \log \tau)$  tiles (shown in green in the zoom-in in Figure~\ref{fig:strongSimOneSeedOverview_seed}) that encode the 4 glues of $t$, and their strengths, in binary. Encoded tile types are separated by a single marker tile shown in blue in  Figure~\ref{fig:strongSimOneSeedOverview_seed}.

The initial configuration of the simulator contains an infinite number of copies of the comb. Via the self-assembly process, each copy of the comb chooses to encode a tile type $t \in T_{\mathrm{in}}$, as follows.  Using tile types from $U_\tau$,  growth initiates from the lower left corner of the comb as shown in part (1) of Figure~\ref{fig:strongSimOneSeedOverview}.   An  $O(\log |G| + \log \tau)$ width crawler is initiated.\footnote{The crawlers, counters, computational primitives (guessing strings, computing simple numerical functions on bit strings, and even simulating Turing machines), and geometric primitives (copying bit sequences around in two-dimensional space) used in this and later constructions are relatively straightforward implementations similar to those used in the aTAM in \cite{IUSA}, among others.  These primitives are designed to assemble on the edges of existing supertiles (or assemblies in the aTAM), and can be made (and usually already are) ``2HAM-safe'' (essentially, ``polyomino safe'' as in \cite{Luhrs08}), meaning that in the 2HAM they function identically and correctly without danger of unwanted supertiles forming which are unattached to the desired supertiles.  The general technique is to limit the number of $\tau$-strength glues on any particular tile type which assembles the primitive to $1$, so that the largest unattached supertile which can form from them is a size $2$ duple.  All other attachments, and even the incorporation of the duples, requires cooperation provided by the surface of the supertile onto which the primitive is intended to form.  Since the constructions for these primitives are standard and straightforward, we omit the details here.}
  This crawler ``guesses'' which tile type the comb should represent via a nondeterministic procedure (in a way that guarantees that any $t\in T_{\mathrm{in}}$ can be guessed), which works as follows. Via nondeterministic placement of $O(\log |G|)$ binary (0 or 1) tiles, the crawler guesses two bit strings, representing positive integers $x,y$. This pair of integers will act as indices to a location on the two-dimensional comb. The crawler climbs to the north of the leftmost tooth, as shown in Figure~\ref{fig:strongSimOneSeedOverview}(2), and then crawls to the east, counting exactly $x$ teeth (3). The crawler heads south (4), counting $y$ blue marker tiles (shown in Figure~\ref{fig:strongSimOneSeedOverview_seed}: zoom-in). At the $y^{\mathrm{th}}$ marker, the crawler has found the encoding of its chosen tile type $t \in T_{\mathrm{in}}$. Here (5), through cooperative binding, the crawler ``reads'' the encoding of $t$ from  the tooth. (The tile type $t$ is encoded as $O(\log |G| + \log \tau)$ tiles, that represent the 4 glues, and their strengths in binary.) The crawler rotates and copies this information to the north (6). Upon reaching the top of the comb, the crawler splits into multiple crawlers, sending each of the 4 glue-strength pairs to their 4 respective sides~(7). In this process each encoded glue~$g$ acts as an index of the position, and length, of the relevant macrotile arm (in other words, an encoded glue $g$ acts as input to a counter that counts to an arm position, and then turns, and counts out an arm length---it can be easily seen that the $O(\log |G|)$ bits used to encode the glue in the crawler are sufficient to store the counter input). An arm generates a green and red glue-binding pad (of the same form as used the proof of in Theorem~\ref{thm:strongSim2}; see Figure~\ref{fig:strongSimArmsDetail}), positioned as shown in Figure~\ref{fig:strongSimOneSeedOverview}.  Note that the green pad with the bumps and dents that encode the glue type in geometry must complete first (and therefore must be two tiles wide so that a path can grow out into each bump and then back down before continuing to the next) before the red pad forms, which is easily done by designing them to assemble as a path completely through green and then to red.  This prevents the situation where the red pad with the generic glues could form first, allowing another macrotile to potentially bind without having the identity of the glue verified by the geometry (meaning that it could allow macrotiles representing mismatched glues of the simulated system to bind). Arm positions and lengths, and the position and structure of glue-binding pads, all follow the form used to prove Theorem~\ref{thm:strongSim2}, although the scaling is different here: specifically, as argued below, the body (outlined in grey in Figure~\ref{fig:strongSimOneSeedOverview}) is of size $k \times k$ where $k \in O( \sqrt{|T|}(\log \tau + \log |G|) + \sqrt{ |G| (\tau + \log |G|) } )$.

This completes the description of (A); the encoding of the {\em singleton} tile types $t \in T$ that are used in the initial state $\sigma$. The remaining supertiles (B) of $\sigma$ (i.e. of size $> 1$) are ``hard-coded'' as supertiles $\hat{s}$ (using tiles from $U_\tau$) that encode supertiles $s \in \sigma$ (analogously to Theorem~\ref{thm:strongSim2}). Essentially these pre-built (hardcoded) supertile assemblies are composed of fully grown comb macrotiles, complete with arms and glue-binding pads.

\paragraph{Scale.}
We first analyze the size of the {\em body} of a filled-out comb macrotile (shaded in grey in Figure~\ref{fig:strongSimOneSeedOverview}).  $T_\mathrm{in} \subseteq T$, so $|T_\mathrm{in}| \leq |T|$.  Each of the $\leq \lceil\sqrt{|T|}\rceil$ teeth of the comb encodes $\leq \lceil\sqrt{|T|}\rceil$ tile types from~$T_{\mathrm{in}}$. Each simulated tile $t \in T_{\mathrm{in}}$ is encoded using $O(\log |G| + \log \tau)$  tiles (shown in green in the zoom-in in Figure~\ref{fig:strongSimOneSeedOverview_seed}) that encode the 4 glues of $t$, and their strengths, in binary. Crawlers are of width $O(\log |G| + \log \tau)$. Together these terms sum to $\ell_1  = O( \sqrt{|T|}(\log \tau + \log |G|) )$, giving the $\ell_1 \times \ell_1$ scaling for the body (shaded in grey in Figure~\ref{fig:strongSimOneSeedOverview}).

The square arm regions for a filled-out comb macrotile  use the same values as appeared in Theorem~\ref{thm:strongSim2}: the glue-binding pad is of size $O(\tau + \log |G|)$, there are $|G|$ such pads, which are rasterized into an $\ell_2 \times \ell_2$ region where $\ell_2 = \sqrt{ |G| (\tau + \log |G|) } $.

We take  $\max(\ell_1 , \ell_2) = k$ to get our final scaling of $k \times k$ where  $$k = O( \sqrt{|T|}(\log \tau + \log |G|) + \sqrt{ |G| (\tau + \log |G|) } ) \, . $$

\paragraph{Correctness of Simulation.}
The simulation correctness uses the same argument as in the proof of Theorem~\ref{thm:strongSim2}, along with the following observations.  In Theorem~\ref{thm:strongSim2} the macrotiles were ``hardcoded'' in advance and so no actual growth takes place within the macro tile itself. Here, for Theorem~\ref{thm:strongSim3}, the macrotiles grow from a comb. Macrotiles bind via the red binding pads (Figure~\ref{fig:strongSimOneSeedOverview}). First, observe that after a comb has selected the values $x,y$, its (future) identity $\hat{t}$ is completely determined. Second, observe that macrotiles bind to each other via their red binding pads only. Third, from the above description of the construction, the red binding pads form only {\em after}  the green glue pads. These three facts together, imply that macrotiles  bind in the simulator if and only if they bind in the simulated system. From here, the correctness argument proceeds as in the proof of Theorem~\ref{thm:strongSim2}.

%% file: sqrtLogWeakSimulation.tex

\section{The temperature-$\tau$ 2HAM is intrinsically universal: simulation}\label{sec:weakSimulation}

In this section we give three results that exhibit for any integer $\tau \geq 2$, a single set of tiles $U_\tau$ that can simulate any temperature $\tau$ 2HAM system at temperature $\tau$ given a proper configuration of initial assemblies over $U_\tau$.  Of particular focus in this section is the achievement of logarithmic scale factors in the size of the simulated tile system.  Beyond this, each of the three results in this section depicts different tradeoffs between the size of $U_\tau$ in terms of $\tau$ versus the scale of the simulation in terms of $\tau$.

As in Section~\ref{sec:strongSimulation}, we adopt the following notation. Let $|| M ||$ denote the number of distinct elements in the multiset $M$, i.e.\ $|| M || = |M'|$ is the cardinality of the set $M'$ induced by the multiset $M$.  Let $T_{\mathrm{sup}}$ denote the set of supertiles induced by a tile set $T$. By this, we simply mean $T_{\mathrm{sup}}$ is the set of supertiles formed by taking all tiles in $T$ and translating them  to all locations in $\mathrm{Z}^2$.

\begin{theorem}\label{thm:weak1}
For each integer $\tau \geq 2$ there exists a single set of tile types $U_{\tau}$, $|U_{\tau}| = O(\tau)$, such that for any 2HAM system $\mathcal{T} = (T,\sigma,\tau)$, there exists a set $I_\mathcal{T}$ of $O(|T|+||\sigma||)$ input supertiles  such that the 2HAM system $\mathcal{U_T}=(U_\tau , U_{\tau,\mathrm{sup}} \bigcup I_\mathcal{T},  \tau)$ simulates $\mathcal{T}$ at scale $O( \sqrt{\log |T|} )$.
\end{theorem}

\begin{figure}[htp]
\begin{center}
\includegraphics[width=\textwidth]{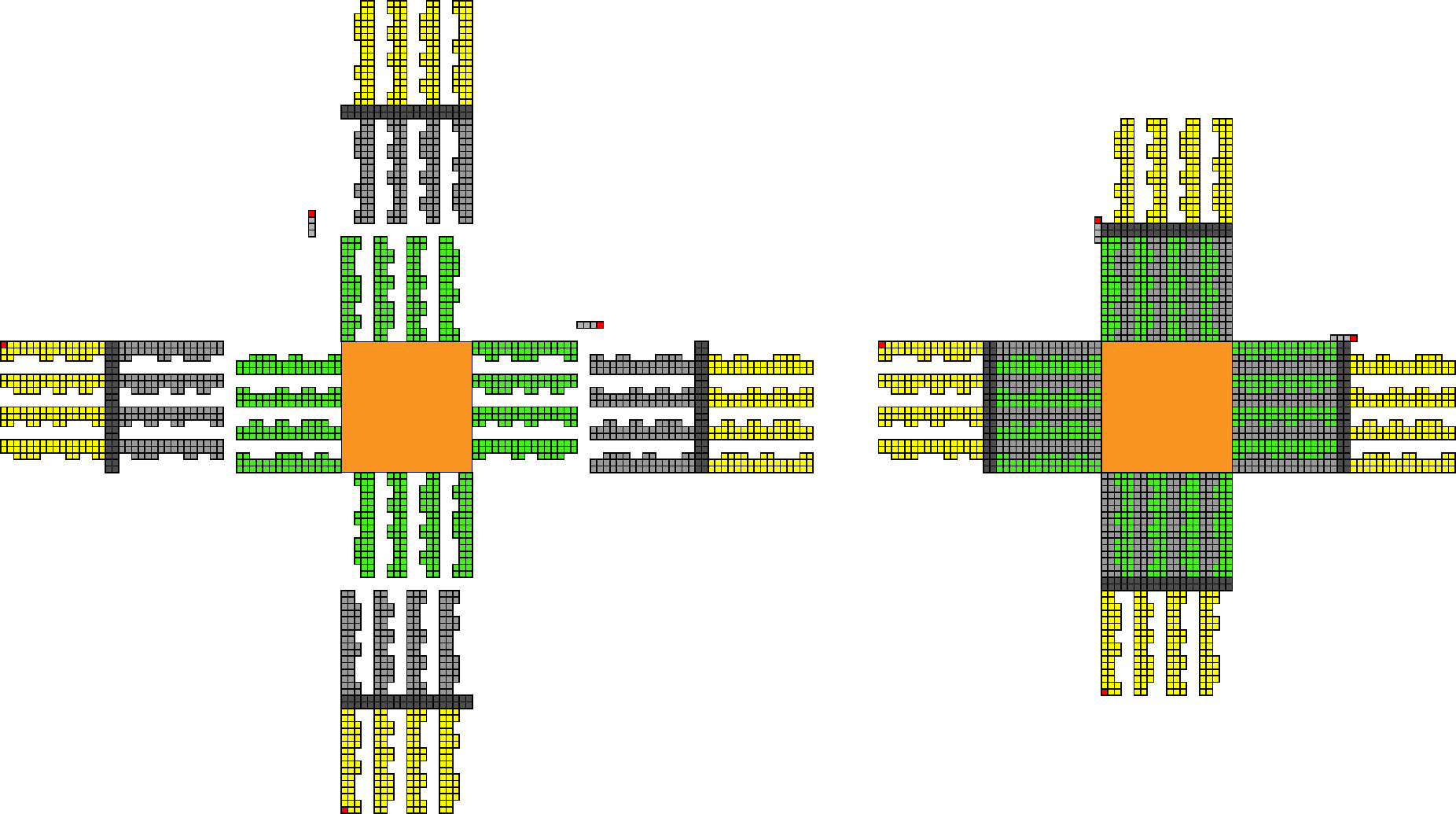}
\caption{These input assemblies for each tile type $t$ comprise the construction for Theorem~\ref{thm:weak1} and are each termed the \emph{megatile} representation of their respective tile $t$.}
\label{fig:sqrtLogInputBlocksMacroManyGlues}
\end{center}
\end{figure}

\paragraph{Construction.}
For a given tile system $\mathcal{T} = (T, \sigma,\tau)$, we first construct a collection of supertiles for each tile in $T$ corresponding to the supertiles depicted in Figure~\ref{fig:sqrtLogInputBlocksMacroManyGlues}.  We then construct an assembly for each element of $\sigma$ by placing combined instances of the assemblies for tile types according the tiles making up the assemblies from $\sigma$.  First, for each tile in $T$, we construct a corresponding \emph{macrotile} consisting of a scale $4 \lceil \sqrt{\log{|T|}} \rceil$ orange square assembly with protruding green \emph{teeth} growths whose geometry encodes a unique identifying binary number for the tile represented by the macrotile.  For each set of green teeth on each of the four sides of the macrotile, there is also a matching \emph{glue gadget} assembly shown with grey and yellow tiles.  The grey portion of this gadget encodes the complement of the macrotile's green geometry, allowing for a snug fit of the appropriate glue gadget to each face of the macrotile.  Additionally, some $\tau$ strength glue is exposed on some portion of the green teeth that matches a glue on the grey portion of each glue gadget (note that all glue gadgets share this same glue and thus rely on the unique complementary geometry to enforce that the unique correct glue gadget is the only gadget that may attach to the macrotile face).  The yellow portion of the glue gadget encodes in geometry a unique binary string representing the glue that occurs on the respective face of the tile to be simulated.  Further, a given glue type represented by a gadget on an east face exposes a binary bump pattern that is the complement to the glue gadget for the same glue type occurring on any west face.  The same complementary setup is used for north and south glue gadgets.  A close-up of the green teeth portion of the assembly, and the attachable glue gadget, is shown in Figure~\ref{fig:sqrtLogInputGlueGadgets}.

\begin{figure}[htp]
\begin{center}
\includegraphics[width=0.6\textwidth]{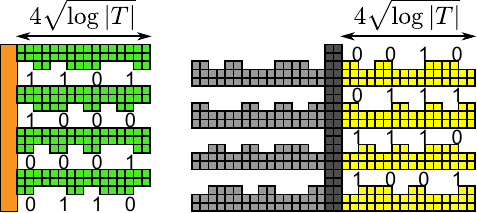}
\caption{Each tile type can select the proper glue gadget by exposing a pattern of green binary teeth unique to the tile type.  Each glue gadget exposes a binary sequence of yellow teeth encoding the glue type represented by the glue gadget.}
\label{fig:sqrtLogInputGlueGadgets}
\end{center}
\end{figure}

In addition to the macrotiles for each tile in $T$, we also have a supertile consisting of a collection of macrotiles for each supertile in $S \in \sigma$.  The construction is the natural one in which a copy of the representing macrotile is placed, at scale, for each of the corresponding tiles in $S$.  For any adjacent tile faces in $S$ that match glue type, the glue gadgets for the respective sides of the respective macrotiles are attaching, along with the glue assemblies, implying that the two macrotiles are connected to each other with the same net strength that the tiles in $S$ are connected with.  For mismatched adjacent glues, both glue gadgets are unattached (or, alternately, just one of the two glue gadgets is unattached), which ensures that the macro blocks can sit together without occupying the same tile positions.  For all glue faces in $S$ that are exposed, i.e., are not adjacent to another tile face, the corresponding macrotile edge does not have its glue gadget attached.  This initial lack of glue gadgets is fundamentally important to ensure a proper simulation in the case of glue mismatches.

Given this set of supertiles $I_{\mathcal{T}}$ derived from $T$ and $\sigma$, we obtain the system $\mathcal{U_{T}} = (U_\tau , U_{\tau,\mathrm{sup}} \bigcup I_{\mathcal{T}}, \tau)$ where $U_\tau$ is a generic set of tiles used to form the body of the construction's macrotiles, as well as a set of $\tau$ tiles used to supply a set of tiles that form the exposed red glues of strength from 1 to $\tau$, thereby yielding a size $|U_\tau| = O(\tau)$ tile set for the construction.  Note that $U_\tau$ only depends on $\tau$ and can be used regardless of which system is being simulated.

Assembly for system $\mathcal{U_{T}} = (U_\tau , U_{\tau, \mathrm{sup}} \bigcup I_{\mathcal{T}}, \tau)$ occurs via attachment between a base macrotile's exposed macro face and any of the macrotile's corresponding glue gadgets.  Upon attachment of either a north or east glue gadget, a cooperative binding site is exposed allowing for the attachment of the 4-tile assembly shown as 3 grey tiles and a red tile in Figure~\ref{fig:sqrtLogInputBlocksMacroManyGlues}.  In the case that the strength of the glue represented by the corresponding glue gadget is less than $\tau$, this 4-tile assembly attaches a glue (exposed on the surface of the red tile) of strength equal to the strength of that glue.    In the case that the glue to be represented is a $\tau$-strength glue, the red glue has strength only $\tau-1$ and the glue gadget is assumed to expose an additional strength-1 glue somewhere along its surface.  The exposed red glues on north and east macrotile glue gadgets match (any of) the red glues of equal strength on the west and south glue gadgets.  Again, since the same glue is used for all distinct glues of the same strength, the geometric compatibility of the binary teeth of the gadget are relied upon to ensure only complementary glues will realize this attraction.

\paragraph{Scale and Tile Set Size.} The scale of this construction is $O(\sqrt{ \log |T| })$ with the mapping from macrotile assemblies to simulated assemblies being the natural mapping implied by the mapping from each macrotile to the tile type the macrotile is derived from.  The bulk of the assemblies can be constructed from an $O(1)$ set of tile types.  An additional special set of $\tau$ tile types is  required to place exposed glues of strengths in the range of $1$ to $\tau$.  Instances of the special set of tiles are represented in the figure by red tiles.

\paragraph{Correctness of Simulation.}  The mapping of supertile blocks to tiles in $T$ is the natural mapping from the large orange blocks to the unique tile in $T$ that the macrotile was designed for.  Whenever two assemblies of macrotiles attach, they must do so based on a sufficient number of matching glue gadgets with a sufficient strength of at least $\tau$.  By design, then, the assemblies mapped to by these two macrotiles have the same exposed glue, and therefore must be able to attach.  Thus, the system $\mathcal{T}$ from which the macrotiles and macro assemblies are derived follows the derived system $\mathcal{U_{T}}$.

We now argue that the derived system $\mathcal{U_{T}}$ weakly models the original system $\mathcal{T}$.  Suppose some assembly $\alpha \in \prodasm{\mathcal{T}}$ can attach to some $\beta \in \prodasm{\mathcal{T}}$ to produce some $c \in \prodasm{\mathcal{T}}$.  Now consider any $\alpha'$ macrotile assembly that maps to $\alpha$.  We must show that $\alpha'$ can grow into an $\alpha''$ that also maps to $\alpha$, and that there must exist a $\beta'$ that can attach to $\alpha''$ to form an assembly that maps to $c$.

The $\alpha''$ that suffices is the super tile consisting of attaching to $\alpha$ all glue gadgets on any exposed macro surfaces of $\alpha$, along with the 4-tile glue assemblies that expose the red tiles with the exposed glues.  The $\beta'$ that suffices is any supertile that maps to $\beta$ with the added restriction that the only glue gadgets that have attached to $\beta'$'s exposed macro edges correspond to the set of glue faces that correspond to matched pairs of glues in the bonding of $\alpha$ and $\beta$.  This ensures that $\alpha'$ and $\beta'$ will not have mutually exclusive \emph{mismatched} glue gadgets that might prevent attachment via geometric hindrance.  Further, the inclusion of all matching glues ensures that enough affinity for attachment will be present based on the assumption that $\alpha$ and $\beta$ are combinable.

\begin{theorem}\label{thm:weak2}
For each integer $\tau \geq 2$ there exists a single set of tile types $U_\tau$, $|U_\tau|=O(1)$ (i.e. independent of $\tau$), such that for all 2HAM systems $\mathcal{T} = (T,\sigma,\tau)$, there exists a set $I_\mathcal{T}$ of $O(|T| + ||\sigma||)$ input supertiles such that the 2HAM system $\mathcal{U_T} = (U_\tau, U_{\tau,\mathrm{sup}} \bigcup I_\mathcal{T},  \tau)$ simulates $\mathcal{T}$ at scale $O( \sqrt{\log |T| + \tau } )$.
\end{theorem}

\begin{figure}[htp]
\begin{center}
\includegraphics[width=\textwidth]{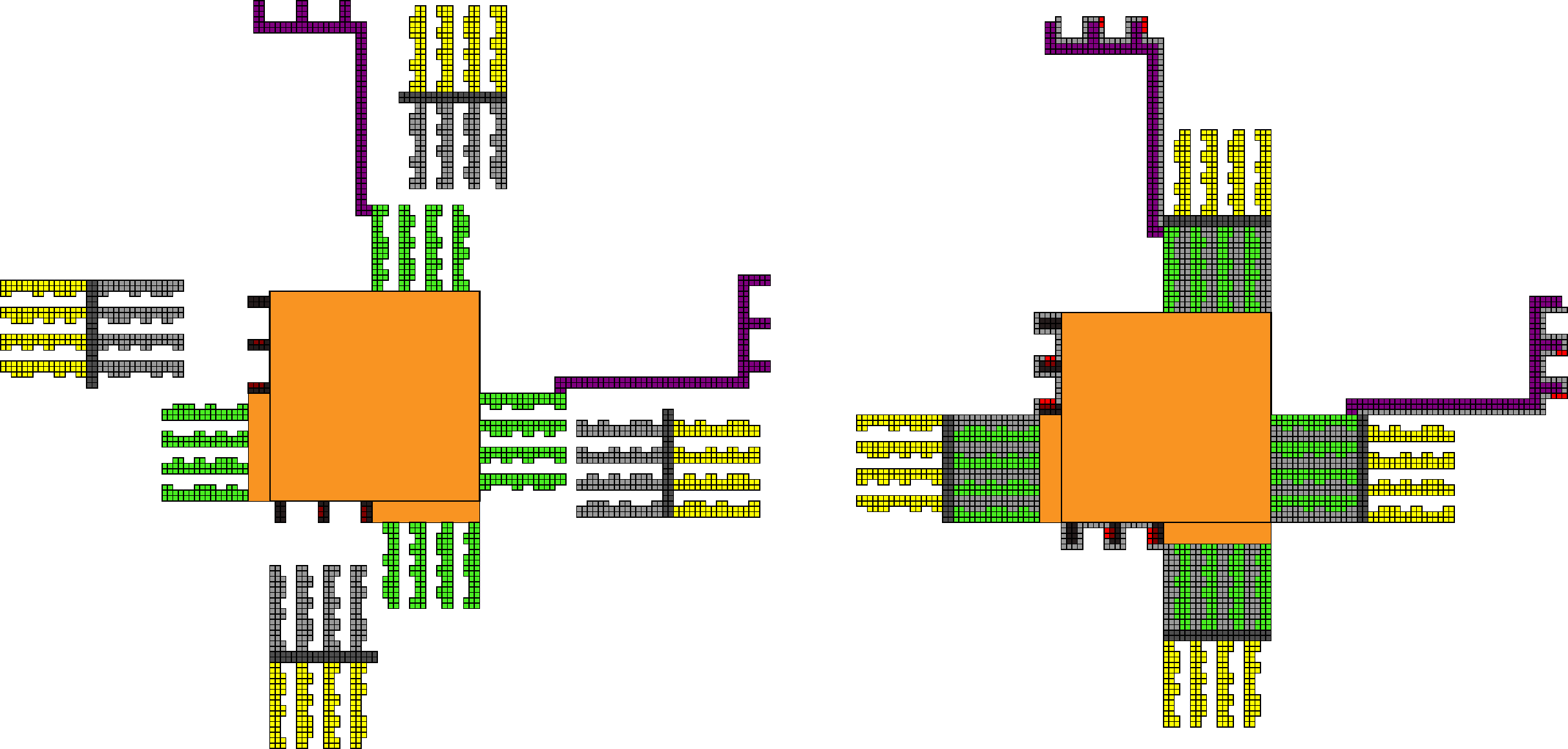}
\caption{To simulate any 2HAM system from a set of tiles $U_\tau$ that does not scale in size with $\tau$, the construction for Theorem~\ref{thm:weak2} extends the construction of Theorem~\ref{thm:weak1} to encode glue strength linearly with single strength glues.  This construction increases the scale by a $\sqrt{\tau}$ factor.  Alternately, Theorem~\ref{thm:weak3} provides a much more modest scale factor increase of $O(\sqrt{\log \tau})$ by encoding strengths through a binary representation and paying the price of a size $|U_\tau| = O(\log \tau)$ set of tiles, which represents a nice compromise between the alternate extremes.
}
\label{fig:sqrtLogInputBlocksMacro}
\end{center}
\end{figure}

\paragraph{Construction.}  The construction for this simulation is a modification of the construction for Theorem~\ref{thm:weak1} and the blocks utilized for each tile in $T$ are shown in Figure~\ref{fig:sqrtLogInputBlocksMacro}.  The key extension in this construction is that the the exposed glues that provide the affinity for attachment among macrotiles with attached glue assemblies can no longer make use of a size $\tau$ set of tiles, but instead must encode $\tau$ through a smaller tile set.  In this construction we utilize a linear encoding of $\tau$ that uses an optimal $O(1)$ set of tiles by stacking single strength glues to create a net force of the desired glue strength.  This linear encoding is displayed in the $\sqrt{\tau} \times \sqrt{\tau}$ rectangular regions of dark tiles on the west and south block faces, and at the end of the extended purple arms on north and east faces.  As with the previous construction, attachment of the glue assemblies exposes a cooperative bonding site, in this case permitting a chain of grey tiles to \emph{coat} the $O(\sqrt{\tau}) \times O(\sqrt{\tau})$ rectangular region with tiles that expose a number of red tiles equal to the strength of the glue represented on the macrotile face.  Each red tile exposes a single strength glue, yielding the desired net force of attraction.

\paragraph{Scale and Tile Set Size.} The scale of this construction is $O(\sqrt{ \log |T| + \tau })$.  Each assembly can be constructed from an $O(1)$ set of tile types as the additional $O(\tau)$ glues types are no longer needed based on the linear encoding of each glue strength.

\paragraph{Correction of Simulation.}  The argument for correct simulation is essentially the same as for Theorem~\ref{thm:weak1}.  A sample figure depicting how the subtle issue of mismatched glues are handled in this simulation is shown in Figure~\ref{fig:sqrtLogAssemblyMismatch}.

\begin{theorem}\label{thm:weak3}
For each integer $\tau \geq 2$ there exists a single set of tile types $U_{\tau}$, $|U_{\tau}| = O(\log \tau)$, such that for any 2HAM system $\mathcal{T} = (T,\sigma,\tau)$, there exists a set $I_\mathcal{T}$ of $O(|T| + ||\sigma||)$ input supertiles such that the 2HAM system $\mathcal{U_T}=(U_\tau , U_{\tau,\mathrm{sup}} \bigcup I_\mathcal{T},  \tau)$ simulates $\mathcal{T}$ at scale $O( \sqrt{\log |T| + \log{\tau} } )$.
\end{theorem}

This theorem represents an interesting compromise between the $O(\tau)$ size tileset required for the simulation for Theorem~\ref{thm:weak1}, and the substantial $\sqrt{\tau}$ scale factor increase from Theorem~\ref{thm:weak2}.  This middle ground approach utilizes a construction similar to Theorem~\ref{thm:weak2}, modified to make use of $O(\log \tau)$ tile types, each of which exposes a glue strength equal to one of the numbers $1, 2, 4, \ldots 2^{\log \tau}$.  The linear encoding of glue strengths via the $\sqrt{\tau} \times \sqrt{\tau}$ rectangular regions from Theorem~\ref{thm:weak2} may now be implemented with smaller $\sqrt{\log \tau} \times \sqrt{\log \tau}$ rectangular regions by using a binary encoding of each glue strength.

\begin{figure}[htp]
\begin{center}
\includegraphics[width=\textwidth]{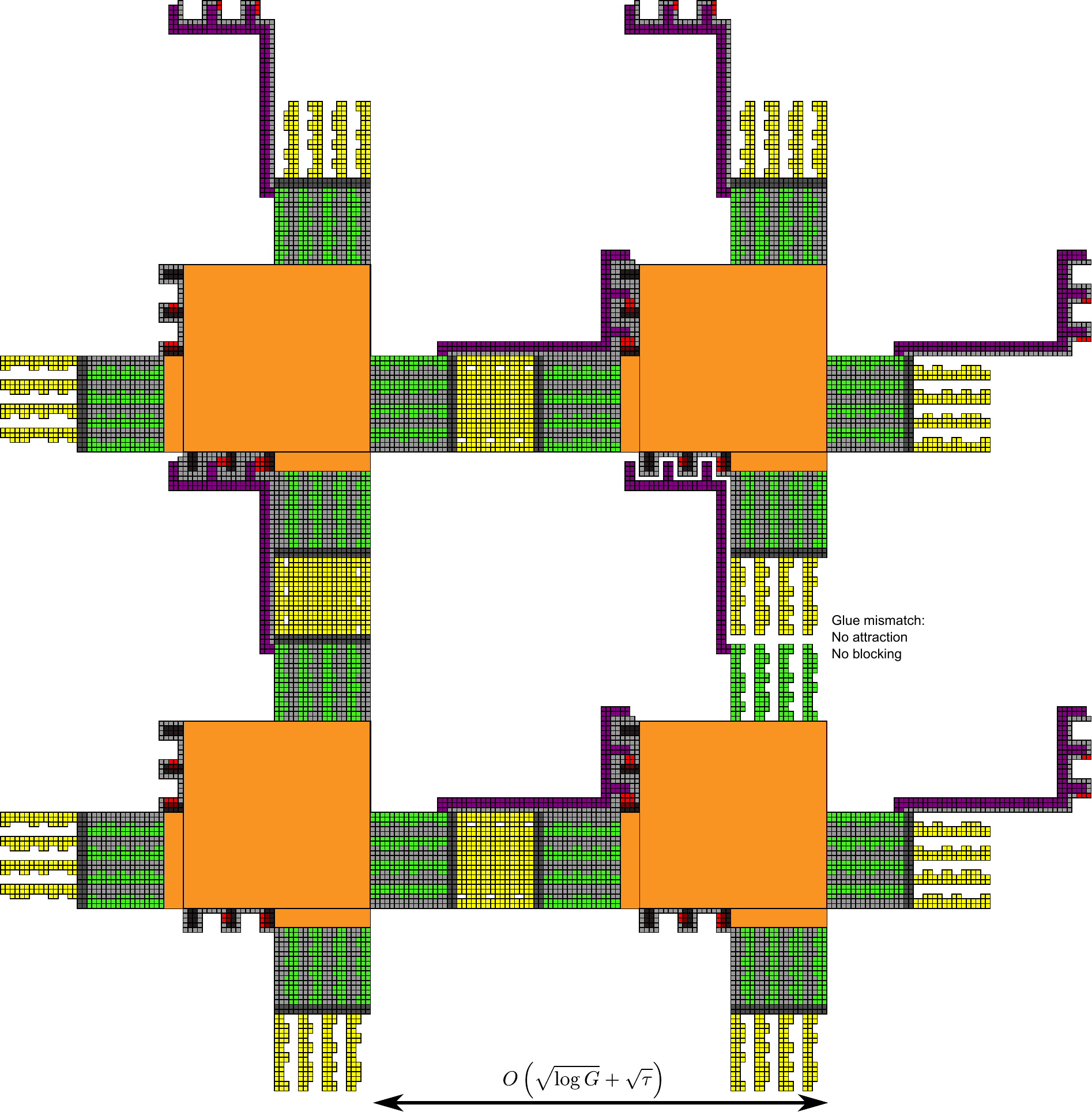}
\caption{Megablocks that have attached glue assemblies and subsequently placed enough single-strength glue tiles may begin attaching to each other based on the combined strength of matching glues.  The geometric teeth of the glue assemblies ensure that only matching glues types will be compatible, thus preventing mismatches from providing positive strength attachments.
Further, let $A'$ be any simulated assembly that combines with simulated assembly $B'$. Then for any assembly $A$ that maps to  $A'$ there must exist a simulator assembly $B$ that  attaches to $A$, even in the presence of simulated glue mismatches  (such a $B$ must exist---it simply has not yet attached its glue assembly in the positions for which simulated  glue mismatches occur).
}
\label{fig:sqrtLogAssemblyMismatch}
\end{center}
\end{figure}

%% file: sim_all.tex

\section{Simultaneous simulation of all 2HAM systems at temperature $\tau$}\label{sec:simAll}

In previous sections, we have shown that while the entire 2HAM is not intrinsically universal, the 2HAM at each temperature $\tau$ is intrinsically universal. In this section we show that there exists a tile set in the 2HAM at each temperature $\tau > 1$ which simultaneously and in parallel simulates \emph{every} 2HAM tile assembly system (with a default initial state) at temperature $\tau$. 
These simultaneous simulations are guaranteed to occur in parallel without any interaction between the macrotiles simulating different systems, and to faithfully strongly simulate each system.  Note that due to technical reasons discussed in Section~\ref{sec:simulate_all_proof}, this construction finitely self-assembles (see \cite{Versus}) the supertiles of the simultaneous simulations.

Before we formally can state our main result for this section, we must define the notion of \emph{simultaneous strong} simulation.

Assume we have a standard enumeration of every 2HAM TAS at temperature $\tau$, i.e., $\mathcal{T}_0 = (T_0, \tau), \mathcal{T}_1 = (T_1, \tau), \ldots\ $. Let $\mathcal{U} = \left(U_{\tau}, \tau\right)$ be a 2HAM TAS and for each $i \in \mathbb{N}$, let $R_i$ be an $m$-block representation function $R_i: B^U_m \dashrightarrow T_i$. Define $\mathcal{T} = \left \langle\mathcal{T}_i \right \rangle_{i=0}^{\infty}$ and $R = \left \langle R_i \right \rangle_{i=0}^{\infty}$.

\begin{definition}\label{scott-defn:simultaneous_alt-equiv-prod}
We say that $\mathcal{U}$ and $\mathcal{T}$ have \emph{simultaneous equivalent productions}, and we write $\mathcal{U} \Leftrightarrow_R^{\infty} \mathcal{T}$ if the following conditions hold:
\begin{enumerate}
    \item \label{scott-defn:simultaneous_simulate:equiv_prod_a} For all $\ta \in \prodasm{U}$, there exists at most one value $i \in \mathbb{N}$ such that $\tilde{R}_i\left(\ta\right) \in \mathcal{A}[\mathcal{T}_i]$.
    \item \label{scott-defn:simultaneous_simulate:equiv_prod_b} For all $i \in \mathbb{N}$, for every $\tb \in \mathcal{A}[\mathcal{T}_i]$, there exists $\ta \in \prodasm{U}$ such that $\tilde{R}_i\left(\ta\right) = \tb$.
    \item For all $\ta \in \prodasm{U}$ there exists at most one value $i \in \mathbb{N}$ such that $\ta$ maps cleanly to $\tilde{R}_i\left(\ta\right)$.
\end{enumerate}
\end{definition}

\begin{definition}\label{scott-defn:simultaneous_alt-equiv-dynamic-t-to-s}
We say that $\mathcal{T}$ \emph{follows} $\mathcal{U}$, and we write $\mathcal{T} \dashv_R^{\infty} \mathcal{U}$ if, for any $\ta, \tb \in \prodasm{\mathcal{U}}$ such that $\ta \rightarrow_{\mathcal{U}}^1 \tb$, there exists at most one value $i \in \mathbb{N}$ such that $\tilde{R}_i(\ta) \rightarrow_{\mathcal{T}_i}^{\leq 1} \tilde{R}_i\left(\tb\right)$.
\end{definition}

\begin{definition}\label{scott-defn:simultaneous_alt-equiv-dyanmic-s-to-t-strong}
We say that $\mathcal{U}$ \emph{strongly models} $\mathcal{T}$, and we write $\mathcal{U} \models^{\infty}_R \mathcal{T}$ if for any $i \in \mathbb{N}$, for any $\ta$, $\tb \in \mathcal{A}[\mathcal{T}_i]$ such that $\tilde{\gamma} \in C^{\tau}_{\ta , \tb}$, then for all $\ta', \tb' \in \prodasm{\mathcal{U}}$ such that $\tilde{R}_i(\ta')=\ta$ and $\tilde{R}_i\left(\tb'\right)=\tb$, it must be that there exist $\ta'', \tb'', \tilde{\gamma}' \in \prodasm{\mathcal{U}}$, such that $\ta' \rightarrow_{\mathcal{U}} \ta''$, $\tb' \rightarrow_{\mathcal{U}} \ta''$, $\tilde{R}_i(\ta'')=\ta$, $\tilde{R}_i\left(\tb''\right)=\tb$, $\tilde{R}_i(\tilde{\gamma}')=\tilde{\gamma}$, and $\tilde{\gamma}' \in C^{\tau}_{\ta'', \tb''}$.
\end{definition}

\begin{definition}\label{scott-defn:simultaneous_alt-strong-simulate}
Let $\mathcal{U} \Leftrightarrow_R^{\infty} \mathcal{T}$ and $\mathcal{T} \dashv_R^{\infty} \mathcal{U}$. We say that $\mathcal{U}$ \emph{strongly simultaneously simulates} $\mathcal{T}$ if $\mathcal{U} \models_R^{\infty} \mathcal{T}$.
\end{definition}


\begin{theorem}\label{thm:ham-for-all}
For each $\tau > 1$, there exists a 2HAM system $\mathcal{S} = (U_{\tau},\tau)$ which strongly simultaneously simulates all 2HAM systems $\mathcal{T} = (T,\tau)$.
\end{theorem}

In \cite{Versus}, it was shown that for every temperature $\tau > 1$, for every aTAM system $\mathcal{T} = (T,\sigma,\tau)$ with $|\sigma| = 1$, there exists a 2HAM system $\mathcal{S} = (S,\tau)$ which simulates it.  From this and Theorem~\ref{thm:ham-for-all}, the following corollary arises:

\begin{corollary}\label{cor:ham-for-tam}
For every temperature $\tau > 1$, there exists a 2HAM system $\mathcal{S} = (U_{\tau},\tau)$ which simultaneously simulates every aTAM system $\mathcal{T} = (T,\sigma,\tau)$ where $|\sigma| = 1$.
\end{corollary}

\subsection{Construction overview}
This construction works with the singleton tiles of the tile set $U_{\tau}$ being the only contents of the initial state, and with them, in parallel via the process of self-assembly, forming macrotiles such that each macrotile simulates one specific tile type from one specific 2HAM system at temperature $\tau$.  (Note that such a system, since it has a default initial state, i.e. only singleton tiles, can be fully specified by its tile set since the temperature is also given.) This is done in such a way that every tile type of every temperature $\tau$ 2HAM tile set is represented by a unique macrotile.  These macrotiles are guaranteed to interact only with macrotiles representing tile types from the same tile set, and to do so in such a way that the group of macrotiles for that system strongly simulate that system.  The scale factor of the simulation of each system, and thus the size of the macrotiles used to simulate it, is potentially unique, and depends upon the running time of a Turing machine and the size of the simulated tile set.

Each macrotile construction begins by randomly selecting the tile set to which the macrotile will belong, and then randomly selecting which tile in the chosen tile set that macrotile will simulate. We guarantee that, in parallel, macrotiles strongly simulate every 2HAM tile assembly system at temperature $\tau$.  Thus it is required that each tile type of each tile set has a corresponding macrotile.  Also, each macrotile must encode information about not only the tile type it is simulating, but also must encode information about the tile set (namely, its number in a fixed enumeration) to which the tile type belongs.  Without the information about which tile set the macrotile is simulating, it would be possible for two macrotiles that simulate tiles from different tile sets to bind with strength greater than $0$, possibly violating the definition of strong simulation.

\begin{figure}[htp]
\begin{center}
\includegraphics[width=2.5in]{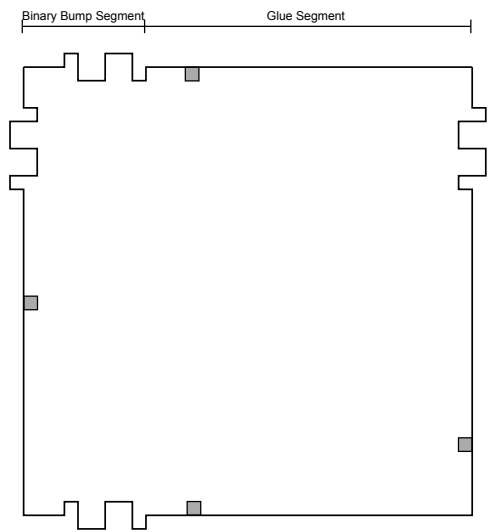}
\caption{A high level overview of a macrotile.}
\label{fig:sim_all_supertile_over}
\end{center}
\end{figure}

\begin{figure}[htp]
\begin{center}
\includegraphics[width=\textwidth]{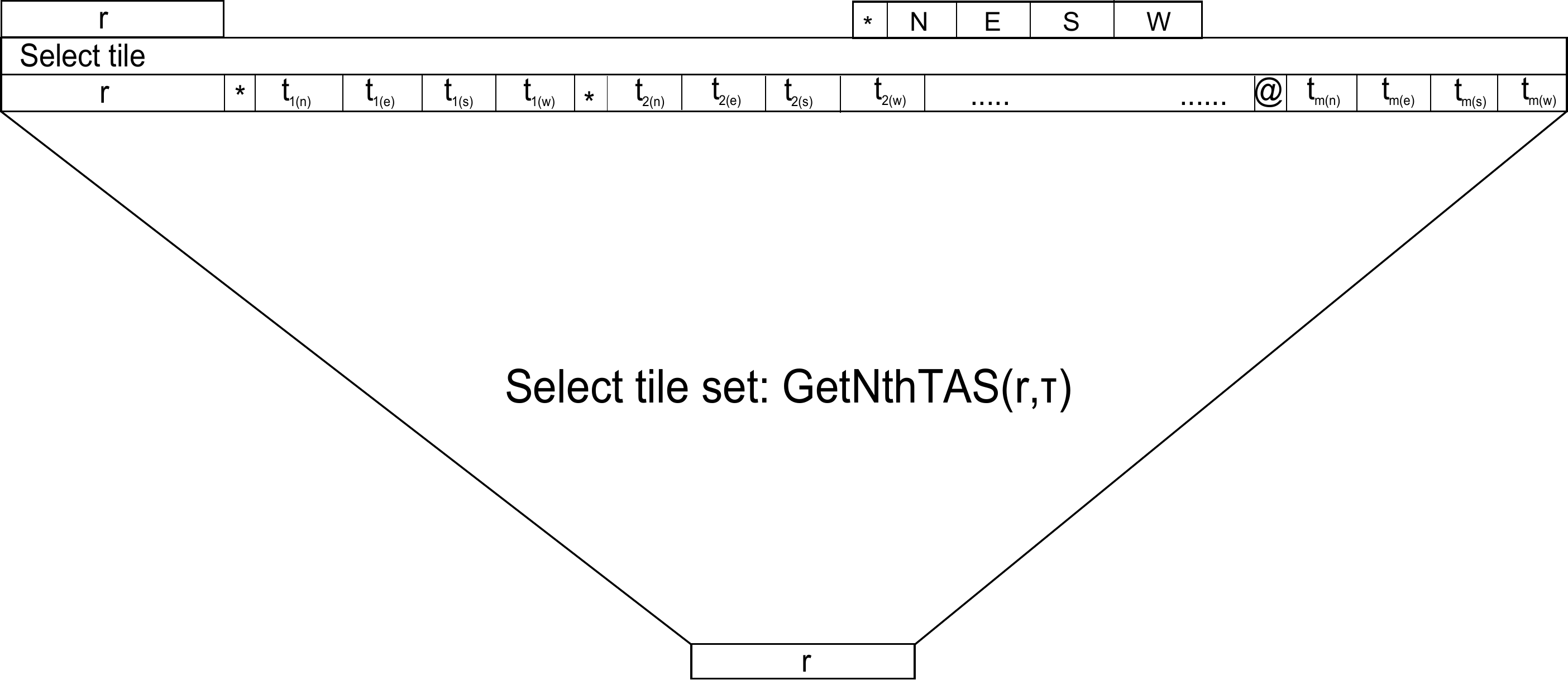}
\caption{A high level overview of the main components involved in building a macrotile. Note that each square represents a segment of tiles rather than an individual tile.}
\label{fig:sim_all_overview}
\end{center}
\end{figure}

\subsection{Construction details}

\begin{figure}[htp]
\begin{center}
\includegraphics[width=2.5in]{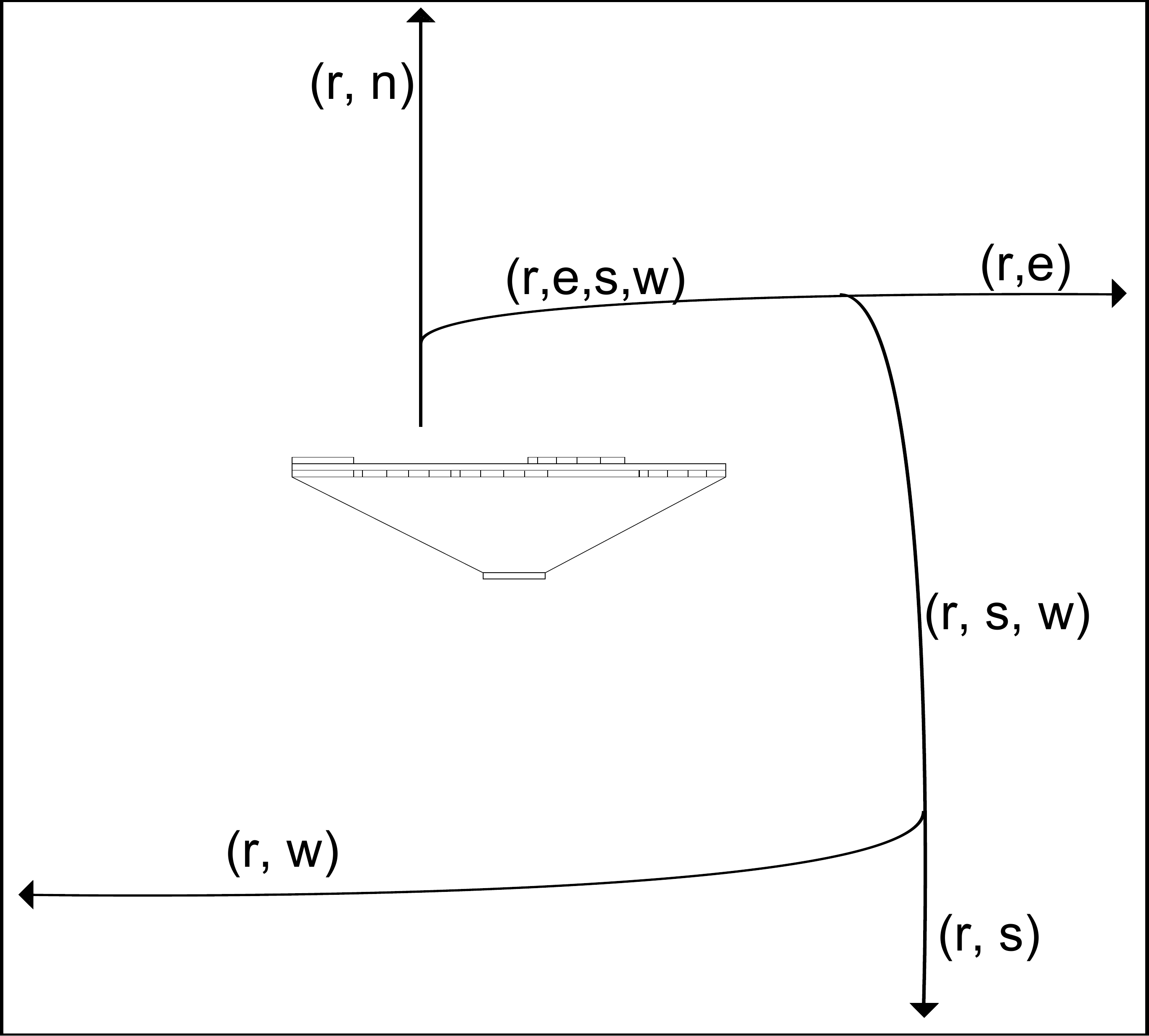}
\caption{The information flow during the assembly of a macrotile.  In this figure, $r$ represents the random number, and n, e, s, w stand for the binary representation of the north, east, south and west glues respectively. }
\label{fig:sim_all_info_flow}
\end{center}
\end{figure}

The assembly of a macrotile begins with the four tile types shown in Figure~\ref{fig:sim_all_r}.  These tiles nondeterministically assemble to yield a random number $r \in \mathbb{Z}^+$.  Next, a standard Turing machine simulation reads $r$ and outputs the tile set with that index number by running the enumeration program described below.  
The tile set is output as a row which encodes each tile type side by side, with spacing tiles in between, and with each tile type represented as the set of its four glues (which in turn are each represented as a pair of binary values for the label and strength of that glue).  Next, exactly one tile type is nondeterministically selected and its information is copied upward. (See Figure~\ref{fig:sim_all_overview}.)  Then, the assembly carries all of the information necessary for each macrotile side to the appropriate locations as seen in Figure~\ref{fig:sim_all_info_flow}.  Next, each macrotile edge first forms a region of bumps and dents which correspond to a binary encoding of the number $r$, which represents the number of the tile set being simulated.  Once those are complete, the remaining portion of the side, which encodes the glue information, forms. (See Figure~\ref{fig:sim_all_supertile_over}.)

\begin{figure}[t] 
\begin{center}
    \subfloat[][Two macrotiles (shown on right) simulating two tiles (on left).]{%
        \label{fig:sim_all_come_together}%
        \centering
        \includegraphics[width=2.0in]{images/sim_all_come_together}
       }%
    \quad\quad
    \subfloat[][Two macrotiles that are not part of the same tile set.  ``Blocking'' caused by the binary teeth representing the different values of $r$, i.e. the tile set numbers, prevents them from binding to one another.]{%
        \label{fig:sim_all_no_together}%
        \centering
        \includegraphics[width=1.8in]{images/sim_all_no_together}
        }
\caption{}
\label{fig:sim_all_examples}
\end{center}
\end{figure}

\subsubsection{Creation of random numbers}
The assembly generates random numbers using the tile set shown in Figure~\ref{fig:sim_all_r}.  Since all of the interior glues are of the same type and $\tau$ strength, these tiles assemble nondeterministically.  Because of this nondeterministic assembly, for any positive integer there exists an assembly of the four tiles such that the assembly is a binary representation of that number.  That is, these four tile types will generate the set of all positive integers.

\begin{figure}[htp]
\begin{center}
\includegraphics[width=0.45\linewidth]{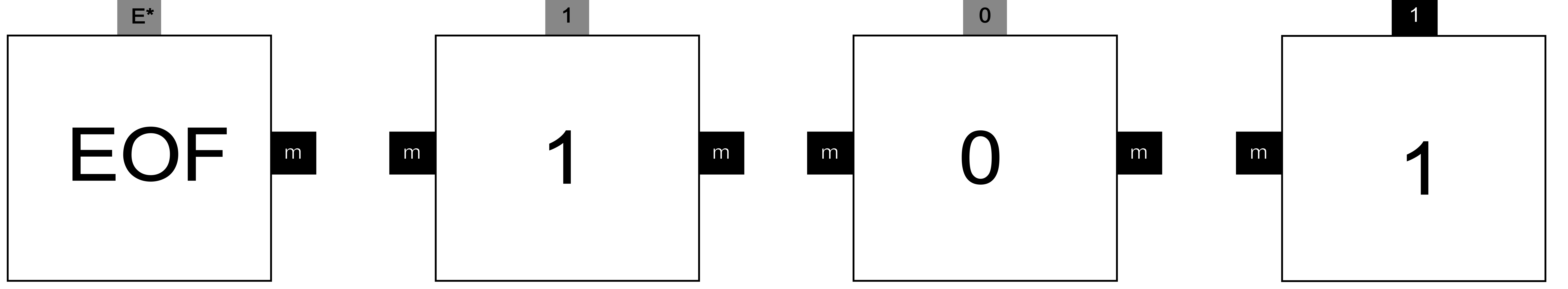}
\caption{The four tiles that bond nondeterministically to generate the random numbers required for our construction. The black squares represent $\tau$ strength glues and the grey squares represent $\lceil {\tau\over2} \rceil$ strength glues.}
\label{fig:sim_all_r}
\end{center}
\end{figure}

\subsubsection{Selecting a tile set}
Now that we have a random number, the assembly can decide to which tile set the macrotile belongs.  In order to accomplish this, we use the program $\proc{GetNthTAS}(tas\_num, \tau)$.  The idea of $\proc{GetNthTAS}(tas\_num, \tau)$ is to enumerate all of the tile sets at a temperature $\tau$.  (See Section~\ref{sec:simulate_all_proof} for more details.)  The algorithm will take a natural number and $\tau$ as input and output the tile set associated with that number.  It does so by, for every possible number of glues (in succession) and combinations of assignments of strengths to those glues from $1$ to $\tau$, creating all possible permutations of unique tile types for those glues.  For each such set, it then enumerates the power set, each element of that being a tile set. The program continues enumerating tile sets until the number assigned to a tile set equals $tas\_num$.  Once this happens, it outputs the tile set associated with this number.  The tile set is outputted as the glues that compose each tile in the tile set with special symbols in between each tile definition.

\begin{codebox}
\Procname{$\proc{GetNthTAS}(tas\_num, \tau)$}
\li    $tile\_set\_number \gets 0$
\li    $|G| \gets 1$
\li    \Do
\li         $glue\_config \gets$ the value of $|G|$ 1's in base $(\tau+1)$
\li         \While $glue\_config \leq (\tau+1)^|G|-1$
\zi         (the value of $|G|$ $\tau$'s in base $(\tau+1)$)
\li         \Do
\li              $a_i \gets$ the $i^{th}$ digit from the least
\zi              significant digit of $glue\_config$
\li              $T = \{\}$
\li              $n \gets 1$
\li               \While $n \leq (|G|+1)^4 -1$
\li               \Do
\li                    $s \gets n$, represented in base $(|G|+1)$ and
\zi                    padded to length 4
\li                    $s_i \gets$ the $i^{th}$ digit from the least significant
\zi                    digit of $s$
\li                    $t \gets ((s_3, a_{s_3}), (s_2, a_{s_2}), (s_1, a_{s_1}), (s_0, a_{s_0}))$
\li                    $T \gets T\cup \{t\}$
\li               \End
\li               \For each $p \in \mathcal{P}(T)$
\li               \Do
\li                 \If $tile\_set\_number = tas\_num$
\li                 \Then
\li                    \Return $p$
\li                 \Else
\li                    $tile\_set\_number++$
\li                 \End
\li             \End
\li           $glue\_config++$
\li       \End
\li     $|G|++$
\end{codebox}

\subsubsection{Selecting a tile}\label{sec:simulate_all_select}
Now that $\proc{GetNthTAS}(tas\_num, \tau)$ has selected the tile set in which our macrotile lives, the assembly chooses the specific tile that this macrotile will become from that tile set.  It achieves this by growing a row of tiles across the top of the row which defines the tile set (which is shown in Figure~\ref{fig:sim_all_overview}).  Until it has selected a particular tile type as that which will be simulated, at each position denoting the beginning of the definition of a new tile type, two tiles are able to bind - one which selects that tile and one which passes on it.  If a tile type is selected, the row which then grows above its definition copies the definition of the tile type upward.  If it is not selected, that information is not propagated upward.  Once a tile type has been selected, all others are ignored (i.e. the tile type which previously could have bound to select the tile can no longer do so).  Finally, if the row grows to the position of the final tile in the tile set (which is specially marked to denote that it is the last tile) and has yet to select a tile type, it is forced to choose this final type.  Thus, every tile type has some probability of being selected, and it is impossible for no tile type to be selected.

\subsubsection{Assembling the macrotile}
As in Section~\ref{subsec:strongSim3Sec}, we are able to use computational and geometric primitives to disperse the information to all sides of the macrotile as shown in Figure~\ref{fig:sim_all_info_flow}.  The glue portion is simply a series of $|G|\tau$ (where $|G|$ is the number of glues in the tile set being simulated) tile locations such that the first $\tau$ positions represent glue $1$, the second $\tau$ positions glue $2$, and so on.  If a given side of a supertile is to simulate the $i$th glue, whose strength is $j$, exactly $j$ of the $\tau$ tiles which represent that glue expose strength $1$ glues to the exterior of the macrotile.  (On the north and south sides of a macrotile it is the westernmost $j$ tiles of the group, and for the east and west sides it is the southern most $j$ tiles of the group.)  All other tiles on that side of the macrotile expose $0$ strength glues.  By forming the binary bumps and dents before the section encoding the glue, it is ensured that only two macrotiles belonging to the same tile set may come together, since only they will have the same pattern and thus geometrically be able to fit together. See Figure~\ref{fig:sim_all_come_together} for an example of how the ``binary bump'' segment allows tiles from the same tile set to come together, and see Figure~\ref{fig:sim_all_no_together} for an example of the ``binary bump'' segment blocking two tiles from coming together that belong to different tile sets.  On the other hand, the glue segment ensures that two macrotiles bind only if the tiles they are representing are able to bind.  As in Section~\ref{subsec:strongSim3Sec}, we must ensure that the ``binary bump'' segment of the side of a macrotile is assemble before the glue segment.  Without this measure, it would be possible for a partially assembled macrotile that has its glue segment assembled but not its teeth to bind to a macrotile simulating a different tile set.

\begin{figure}[htp]
\begin{center}
\includegraphics[width=4.5in]{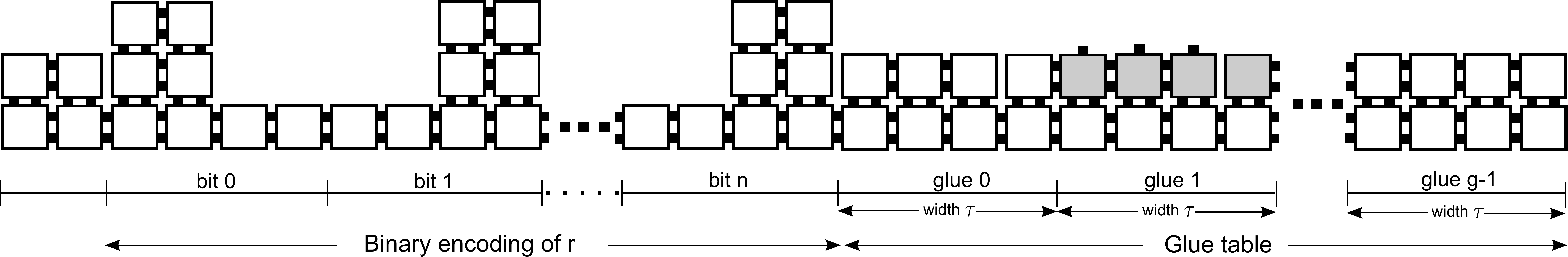}
\caption{Example macrotile side showing the binary teeth which encode $r$ as well as the glue table, which is simply a $\tau \cdot |G|$ (where $|G|$ is the number of glues in system $r$) row of tiles such that the $i$th consecutive group of $\tau$ tiles represent the $i$th glue.  Since each macrotile side can simulate exactly one glue, only the tiles of a single glue segment have non-zero external glues.  Those glues are all strength-$1$ and the number of them is equal to the strength value of the simulated glue. In this example, glue $1$ is the glue being simulated where the strength of glue $1$ is $3$ in $\tau=4$ system.}
\label{fig:sim_all_macrotile_side}
\end{center}
\end{figure}

\subsection{Proof of correctness}\label{sec:simulate_all_proof}
To see that Theorem~\ref{thm:ham-for-all} follows from this construction we show that 1) for every tile type in every possible tile set at temperature $\tau$, a unique macrotile self-assembles which maps to that tile type in that tile set, and 2) each system is strongly simulated by the macrotiles representing it.

We first define what it means for two tile sets to be \emph{functionally equivalent} to each other.

\begin{definition}
We say that tile set $T_1$ is \emph{functionally equivalent} to tile set $T_2$ if there exists a one-to-one mapping function $f: T_1 \rightarrow T_2$ that maps each tile type in $T_1$ to a unique tile type in $T_2$ such that for every pair of tiles $t_a, t_b \in T_1$, if and only if side $d_1 \in \{N,E,S,W\}$ of $t_1$ binds to side $d_2 \in \{N,E,S,W\}$ of $t_2$ with strength $s$, side $d_1$ of $f(t_1)$ binds to side $d_2$ of $f(t_2)$ with strength $s$.
\end{definition}

Therefore, if two tile sets are functionally equivalent, each tile type of each set has exactly one identical counterpart in the other tile set which is able to bind to the same exact sides of the same set of equivalent tiles. Furthermore, we say that two tile sets are \emph{functionally distinct} if they are not functionally equivalent.

We begin with the claim that a set covering all functionally distinct tile sets in the 2HAM at temperature $\tau$ is countable.  To do this we first note that regardless of the specific labels used for the glues of an arbitrary tile set $T$, and in fact there is an uncountably infinite set of possible labels, there is a functionally equivalent tile set which replaces each glue label with the integer $g$ for $0 \le g < |G|$ (represented as a string) where $G$ is the set of glues in $T$. We will refer to such tile sets (i.e. those using the integers as their glue labels) as the \emph{canonical} tile sets.  Now, observe that given a maximum strength (i.e. $\tau$) and a fixed number of glues $|G|$, the tile sets that can be produced is finite.  The number of possible glues for any canonical tile set is given by a positive integer, and the set of positive integers is countable.  It follows that the set of all canonical tile sets in the 2HAM at temperature $\tau$ thus consists of a countable union of finite sets.  Consequently, all the set of all canonical tile sets in the 2HAM at temperature $\tau$ is countable.  Since every tile set is functionally equivalent to a canonical tile set, the countable set of canonical tile sets represents a set of all functionally distinct 2HAM tile sets at temperature~$\tau$.

Now, we note that by the definition of (strong) simulation, the (strong) simulation of a system at temperature $\tau$ which includes a single element of a set of functionally equivalent tile sets $E$ is a (strong) simulation of all systems at temperature $\tau$ which include any tile set in $E$.  Therefore, in order to simulate every 2HAM system at temperature $\tau$, it is sufficient to simulate the set of all systems composed of canonical tile sets.  We now show that $\proc{GetNthTAS}(tas\_num, \tau)$ enumerates every canonical tile set at temperature $\tau$.  For every positive integer $|G|$ in succession, it uses $|G|$ as the number of glues and, using the canonical form of glue labels, creates every possible combination of mappings of glue labels to strengths from $1$ through $\tau$.  For each such mapping, signifying a set of glues and associated strengths, it creates a set of every possible tile type that can be created using those glues (and sides with no glue).  Given this set, the complete set of all tiles that could be created with these glues, it treats each element of the power set as its own tile set, assigning the next tile set number to each, until it reaches the number matching its input value and then outputs that tile set.  Note that, as it iterates and creates power sets, it will create power sets which contain elements (corresponding to tile sets) that are identical to those of other power sets.  This does not pose a problem, as it simply means that such tile sets will be counted multiple times and therefore simultaneously simulated under different numbers.  The important fact is that no canonical tile set is excluded from the enumeration.  In such a way, it is guaranteed that every canonical tile set is included in the enumeration.

Since every tile in a chosen tile set $r$ is selected to be built into a macrotile with some probability $p > 0$ (based on the nondeterministic selection described in Section~\ref{sec:simulate_all_select}) and there will be an infinite number of copies of the assembly signifying that tile set number $r$ has been selected in solution, it follows that there exists a unique macrotile that self-assembles which maps to that tile type for that system.  Note that an assembly sequence exists in which the line assembly built from the tile types used to nondeterministically select $r$ could grow to infinite length (in the limit) and never terminate, preventing the supertile from every growing into a macrotile.  For this reason, we say that this construction finitely self-assembles the macrotiles, as defined in \cite{Versus}.

To show that the behavior of the macrotiles is consistent with the definition of strong simulation, we first define an index function $I: A^{U_{\tau}} \dashrightarrow \mathbb{N}$ which takes as input an assembly over $U_{\tau}$ and returns the index of the tile set for which it is a macrotile simulation of some tile type, or is undefined if the assembly has not yet selected a tile set to simulate.  $I$ can do this by inspecting the supertile to find the number $r$ which was selected (simply looking for tiles from the set in Figure~\ref{fig:sim_all_r}) and running the same Turing machine program used by the construction to find the tile set and scaling factor (which is based on the running time of the Turing machine and the size of the tile set).  Using $I$, we can determine which supertile representation function $R^*_i$ to use to map a supertile over $U_{\tau}$ to a supertile over tile set $T_i$, for $0 \le i < \infty$.  Thus we can separate all supertiles $\alpha \in \prodasm{\mathcal{S}}$ into sets $S_i$ for $0 \le i < \infty$ such that the supertiles in $S_i$ represent macrotiles in the simulation of tile set $T_i$. Now we examine the physical features of macrotiles to prove the claim that those in $S_i$ strongly simulate the system $\mathcal{T}_i = (T_i,\tau)$. Since each macrotile is constructed such that each of its edges are encoded with a ``binary bump'' segment corresponding to $i$ (the tile set it is simulating), it is only possible for macrotiles' glue segments to come into contact if they are both in $S_i$ (i.e. simulating tiles from the same tile set).  Thus for all macrotiles $\alpha \in S_i$ and $\beta \in S_j$ where $i \ne j$, they are simulating different tile sets and their glue segments may never come into contact so they cannot bind with strength $s > 0$. Note that the guaranteed ordering of growth of the features of each side is important for this, because the binary teeth must be in place before the external glues to ensure that correct exclusion will occur for macrotiles from different sets.  Since the glues are placed uniquely on the glue segment with strengths equivalent to the tiles they are simulating, macrotiles bind with strength $s > 0$ if and only if the tiles they are simulating bind with strength $s > 0$.  Thus, by using the corresponding supertile representation function $R^*_i$ for each $S_i$, it is shown that tile set $T_i$ is strongly simulated at temperature $\tau$.

Further, while the macrotiles for a given tile set faithfully strongly simulate that tile set, they do not interfere with any of the (infinite) simulations simultaneously occurring in parallel.  Finally, we note that since the exterior of each macrotile contains no glues of strength $> 1$ along the flat glue portions of each side, it is impossible for the individual tiles of the simulator to attach to a macrotile.  In fact, only once a supertile has grown into a macrotile and represents at least the binary teeth and glue portion of some side, can it interact with any other macrotile.